\documentclass[10pt,twocolumn,letterpaper]{article}
\usepackage[T1]{fontenc}
\usepackage{amssymb}
\usepackage{amsmath}
\usepackage{amsthm}
\usepackage{booktabs}
\usepackage[font={small}]{caption}
\usepackage{cite}
\usepackage{color}
\usepackage{colortbl}
\usepackage{enumitem}
\usepackage{framed}
\usepackage{flushend}
\usepackage[hidelinks]{hyperref}
\usepackage{hhline}
\PassOptionsToPackage{hyphens}{url}
\usepackage{nicefrac}
\usepackage[letterpaper,margin=1in]{geometry}
\usepackage{letltxmacro}
\usepackage{mathtools}
\usepackage{microtype}
\usepackage{multirow}
\usepackage{pbox}
\usepackage{scalefnt}
\usepackage{subfig}
\usepackage{tikz}
\usepackage{txfonts}
\usepackage{url}
\usepackage{usenix}
\usepackage{xcolor}
\usepackage{xspace}

\DeclareMathAlphabet\mathcal{OMS}{cmsy}{m}{n}
\SetMathAlphabet\mathcal{bold}{OMS}{cmsy}{b}{n}

\DeclareSymbolFont{AMSb}{U}{msb}{m}{n}
\DeclareSymbolFontAlphabet{\mathbb}{AMSb}

\definecolor{LightCyan}{rgb}{0.88,1,1}

\setitemize{noitemsep,topsep=2pt,parsep=2pt,partopsep=2pt}

\def\C{\mathcal{C}}

\def\RR{\mathbb{R}}

\newcommand\D{\mathcal{D}}
\newcommand\K{\mathcal{K}}
\newcommand\hatf{\hat{f}}
\newcommand\hatg{\hat{g}}
\newcommand\hath{\hat{h}}

\newcommand\A{\mathcal{A}}
\def\F{\mathbb{F}}

\newcommand{\state}{\mathsf{state}}

\newcommand{\itpara}[1]{\medskip\noindent\textit{#1}}
\newcommand{\sitpara}[1]{\smallskip\noindent\textit{#1}}
\renewcommand{\paragraph}[1]{\medskip\noindent\textbf{#1}}

\newcommand{\afe}{AFE\xspace}
\newcommand{\afes}{AFEs\xspace}

\newcommand{\Enc}{\textsf{Encode}}
\newcommand{\Dec}{\textsf{Decode}}
\newcommand{\Ver}{\textsf{Valid}}

\newcounter{thmc}
\newtheorem{thm}[thmc]{Theorem}

\newtheorem{claim}[thmc]{Claim}

\theoremstyle{definition}
\newtheorem{defn}[thmc]{Definition}

\newcommand{\name}{Prio\xspace}
\newcommand{\Name}{Prio\xspace}

\newcommand{\rulesep}{\unskip\ \textcolor{gray}{\vrule\vrule}\ }
\newcommand{\rulesepThick}{\unskip\ \textcolor{gray}{\vrule\vrule\vrule\vrule}\ }

\makeatletter

\let\c@table\c@figure
\makeatother 

\newcommand{\proj}[1]{\textsf{Trunc}_{#1}}
\newcommand{\Hsh}[1]{[#1]_\textbf{h}}
\newcommand{\Ash}[1]{[#1]_\textbf{a}}

\setlist[description]{leftmargin=\parindent,topsep=0ex,itemsep=0ex,partopsep=1ex,parsep=1ex}

\LetLtxMacro{\oldtextsc}{\textsc}
\renewcommand{\textsc}[1]{\oldtextsc{\scalefont{1.2}#1}}

 \date{\today}
\title{\vspace{-12em}
Prio: Private, Robust, and Scalable Computation of Aggregate Statistics}

\author{Henry Corrigan-Gibbs and Dan Boneh\thanks{\hspace{-2em}
This is the extended version of a paper by the
same title that appeared at NSDI 2017.
This version includes additional security analysis,
corrects typos, and fixes an error in the proceedings version.}\\
Stanford University}

\renewcommand\footnotemark{}

\makeatletter
\begin{document}
\maketitle

\paragraph{Abstract.}
This paper presents \name, a privacy-preserving system for the
collection of aggregate statistics. Each \name client holds a private data value
(e.g., its current location), and a small set of servers compute statistical
functions over the values of all clients (e.g., the most popular location). As
long as at least one server is honest, the \name servers learn nearly nothing about the
clients' private data, except what they can infer from the aggregate statistics
that the system computes. 
To protect functionality in the face of faulty or malicious clients, \name uses 
{\em secret-shared non-interactive proofs} (SNIPs), a new cryptographic technique 
that yields a hundred-fold performance improvement over conventional
zero-knowledge approaches. 
\Name extends classic private aggregation techniques
to enable the collection of a large class of useful statistics. For example,
\name can perform a least-squares regression on high-dimensional client-provided
data without ever seeing the data in the clear.

 \section{Introduction}
Our smartphones, cars, and wearable electronics
are constantly sending telemetry data and other sensor readings
back to cloud services. 
With these data in hand, a cloud service can compute useful {\em aggregate
statistics} over the entire population of devices. 
For example, navigation app providers collect real-time location data from their
users to identify areas of traffic congestion in a city and
route drivers along the least-crowded roads~\cite{jeske2013floating}.
Fitness tracking services collect information on their users' physical activity
so that each user can see how her fitness regimen compares to the average~\cite{fitness}.
Web browser vendors collect lists of unusually popular homepages 
to detect homepage-hijacking adware~\cite{erlingsson2014rappor}.

Even when a cloud service is only interested in 
learning aggregate statistics about its user population as a whole,
such services often end up collecting private data from each client
and storing it for aggregation later on.
These centralized caches of private user data pose severe security and
privacy risks: 
motivated attackers may steal and disclose clients'
sensitive information~\cite{wang2016defending,nyt-hack}, cloud services
may misuse the clients' information for profit~\cite{uber},
and intelligence agencies may appropriate the data for targeting or
mass surveillance purposes~\cite{leaky-apps}. 

To ameliorate these threats, major technology companies, including
Apple~\cite{apple-dp} and Google~\cite{erlingsson2014rappor,fanti2016building},
have deployed privacy-preserving systems for the collection of user data.
These systems use a ``randomized response'' mechanism to achieve
differential privacy~\cite{warner1965randomized,DP}.
For example, a mobile phone vendor may want to learn how many of its phones have
a particular uncommon but sensitive app installed (e.g., the AIDSinfo app~\cite{aids-info}).
In the simplest variant of this approach, 
each phone sends the vendor a bit indicating whether it has the app installed,
except that the phone flips its bit with a fixed probability $p < 0.5$.
By summing a large number of these noisy bits, the vendor can get a good
estimate of the true number of phones that are running the sensitive app.

This technique scales very well and is robust even if some of the phones
are malicious---each phone can influence the final sum by $\pm 1$ at most.
However, randomized-response-based systems provide relatively {\em weak privacy} guarantees:
every bit that each phone transmits leaks some private user information 
to the vendor.
In particular, when $p=0.1$ the vendor has a good chance of seeing the
correct (unflipped) user response. 
Increasing the noise level~$p$ decreases this leakage, but adding more noise 
also decreases the accuracy of the vendor's final estimate.
As an example, assume that the vendor collects randomized responses from one
million phones using $p=0.49$, and that
$1$\% of phones have the sensitive app installed.
Even with such a large number of responses, the vendor will incorrectly
conclude that {\em no phones} have the app installed 
roughly one third of the time.

An alternative approach to the data-collection problem is to have
the phones send {\em encryptions} of their bits to a set of servers.
The servers can sum up the encrypted bits and decrypt 
only the final sum~\cite{melis2016,elahi2014privex,popa2009vpriv,popa2011privacy,joye2013scalable,danezis2013smart}.
As long as all servers do not collude, these encryption-based
systems provide much {\em stronger privacy}
guarantees: the system leaks nothing about a user's private bit to 
the vendor, except what the vendor can infer from the final sum.
By carefully adding structured noise to the final sum, 
these systems can provide differential privacy as
well~\cite{shi2011privacy,melis2016,elahi2014privex}.

However, in gaining this type of privacy, 
many secret-sharing-based systems 
sacrifice {\em robustness}: a malicious client can send the servers an
encryption of a large integer value $v$ instead of a zero/one bit.
Since the client's value $v$ is encrypted, the servers cannot tell from 
inspecting the ciphertext that $v > 1$.
Using this approach, a single malicious client can 
increase the final sum by $v$, instead of by~$1$.
Clients often have an incentive to cheat in this way:
an app developer could use this attack to boost the perceived
popularity of her app, with the goal of getting it
to appear on the app store's home page.
It is possible to protect against these attacks using zero-knowledge proofs~\cite{shi2011privacy},
but these protections destroy {\em scalability}: 
checking the proofs requires heavy public-key cryptographic operations at
the servers and can increase the servers' workload by orders of magnitude. 

In this paper, we introduce \name, a system for private aggregation that
resolves the tension between privacy, robustness, and scalability.
\Name uses a small number of servers;
as long as one of the \name servers
is honest, the system leaks nearly nothing about clients' private data
(in a sense we precisely define), except
what the aggregate statistic itself reveals.
In this sense, \name provides a strong form of cryptographic {\em privacy}.
This property holds even against an adversary who can observe the entire
network, control all but one of the servers, and control a large number 
of clients.

\Name also maintains {\em robustness} in the 
presence of an unbounded number of malicious clients,
since the \name servers can detect and reject syntactically incorrect client
submissions in a privacy-preserving way.
For instance, a car cannot report a speed of 100,000~km/h if the system
parameters only allow speeds between 0~and 200~km/h.
Of course, \name cannot prevent a malicious client from submitting 
an untruthful data value:
for example, a faulty car can always misreport its actual speed.

To provide robustness, \name uses a new technique that we call 
{\em secret-shared non-interactive proofs} (SNIPs).
When a client sends an encoding of its
private data to the \name servers, the client also sends to each
server a ``share'' of a proof of correctness.
Even if the client is malicious and the proof shares are malformed, the
servers can use these shares to collaboratively check---without ever seeing
the client's private data in the clear---that
the client's encoded submission is syntactically valid.  
These proofs rely only upon fast, information-theoretic cryptography, and
require the servers to exchange only a few hundred bytes of information 
to check each client's submission.

\Name provides privacy and robustness without sacrificing {\em scalability}.
When deployed on a collection of five servers spread around the world
and 
configured to compute private sums over vectors of private client data, 
\name imposes a 5.7$\times$ slowdown over a na\"ive data-collection
system that provides no privacy guarantees whatsoever.
In contrast, a state-of-the-art comparison system that 
uses client-generated non-interactive zero-knowledge proofs of correctness
(NIZKs)~\cite{blum1988non,schnorr1991efficient} imposes a 267$\times$
slowdown at the servers.
\Name improves client performance as well: it is 50-100$\times$
faster than NIZKs and we estimate that it is three orders of magnitude faster
than methods based on succinct non-interactive arguments of knowledge
(SNARKs)~\cite{gennaro2013quadratic,ben2013snarks,parno2013pinocchio}.
The system is fast in absolute terms as well:
when configured up to privately collect the distribution of
responses to a survey with 434 true/false questions, the client 
performs only 26~ms of computation, and our distributed cluster of \name
servers can process each client submission in under 2~ms on average.

\paragraph{Contributions.}
In this paper, we:
\begin{itemize}[leftmargin=1\parindent]
  \item introduce {\em secret-shared non-interactive proofs} (SNIPs), 
      a new type of information-theoretic
      zero-knowledge proof, optimized for the client/server setting,
\item present {\em affine-aggregatable encodings}, a framework that unifies
      many data-encoding techniques used in prior work on private aggregation, and
\item demonstrate how to combine these encodings with SNIPs to provide
      robustness and privacy in a large-scale data-collection system.
\end{itemize}

With \name, we demonstrate that data-collection systems can simultaneously
achieve strong privacy, robustness to faulty clients, and performance at scale.

 \section{System goals}
\label{sec:goals}

A \name deployment consists of a small number of infrastructure servers and a
very large number of clients.
In each time epoch, every client $i$ in the system holds a private value $x_i$.
The goal of the system is to allow the servers to compute $f(x_1, \dots, x_n)$,
for some aggregation function $f$, in a way that leaks as little as possible
about each client's private $x_i$ values to the servers.

\paragraph{Threat model.}
The parties to a \name deployment must establish pairwise authenticated
and encrypted channels. 
Towards this end, we assume the existence of 
a public-key infrastructure and the basic
cryptographic primitives (CCA-secure public-key encryption~\cite{cramer1998practical,shoup2001oaep,shoup2001proposal}, 
digital signatures~\cite{goldwasser1988digital}, etc.)
that make secure channels possible.
We make no synchrony assumptions about the network: the adversary may 
drop or reorder packets on the network at will, and
the adversary may monitor all links in the network.
Low-latency anonymity systems, such as Tor~\cite{dingledine2004tor}, provide no
anonymity in this setting,
and \name does not rely on such systems to protect client privacy.

\paragraph{Security properties.}
\Name protects client {\em privacy} as long
as at least one server is honest.
\Name provides {\em robustness} (correctness) only if all servers are honest.
We summarize our security definitions here, but please refer 
to Appendix~\ref{app:secdefs} for details.

\itpara{Anonymity.}
A data-collection scheme maintains client anonymity if the
adversary cannot tell which honest client submitted which data value through
the system, even if the adversary chooses the honest clients' data values,
controls all other clients, and controls all but one server.
\Name always protects client anonymity.

\itpara{Privacy.} 
\Name provides $f$-privacy, for an aggregation function $f$, if
an adversary, who controls any number of clients and all but one server,
learns nothing about the honest clients' values $x_i$, except 
what she can learn from the value $f(x_1, \dots, x_n)$ itself.
More precisely, given $f(x_1, \dots, x_n)$,
every adversary controlling a proper subset of the servers, along with any number of
clients, can simulate its view of the protocol run.

For many of the aggregation functions $f$ that \name implements, \name provides
strict $f$-privacy.
For some aggregation functions, which we highlight in Section~\ref{sec:struct},
\name provides $\hat{f}$-privacy, where $\hat{f}$ is a function that
outputs slightly more information than $f$.  
More precisely, 
$\hat{f}(x_1,\ldots,x_n) = \big\langle f(x_1,\ldots,x_n),\ L(x_1,\ldots,x_n)\big\rangle$
for some modest leakage function $L$.

\Name does not natively provide differential privacy~\cite{DP}, since
the system adds no noise to the aggregate statistics it computes.
In Section~\ref{sec:disc:attacks}, we discuss when differential privacy 
may be useful and how we can extend \name to provide it.

\itpara{Robustness.}
A private aggregation system is robust if
a coalition of malicious clients can affect the output of the system only
by misreporting their private data values;
a coalition of malicious clients cannot otherwise corrupt the system's output.
For example, if the function $f(x_1, \dots, x_n)$ counts the number of times a
certain string appears in the set $\{x_1, \dots, x_n\}$, then a single
malicious client should be able to affect the count by at most one.

\Name is robust only against adversarial clients---{\em not} against adversarial servers.
Although providing robustness against malicious servers seems desirable at
first glance, doing so would come at privacy and performance costs, which we
discuss in Appendix~\ref{app:robfault}.
Since there could be millions of clients in a \name deployment,
and only a handful of servers (in fixed locations with known administrators), 
it may also be possible to eject faulty servers using out-of-band means.

 \section{A simple scheme}
\label{sec:simple}

\newcommand{\imgwidth}{0.23\textwidth}
\begin{figure*}
\centering
\subfloat[The client sends a share of its encoded submission and SNIP proof to each server.]{\includegraphics[width=\imgwidth]{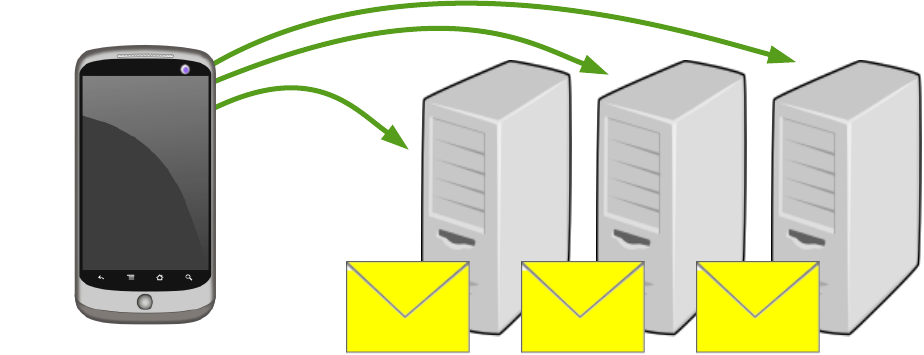}}~~\rulesep{}~\subfloat[The servers validate the client's SNIP proof to ensure that the submission is valid.]{\includegraphics[width=\imgwidth]{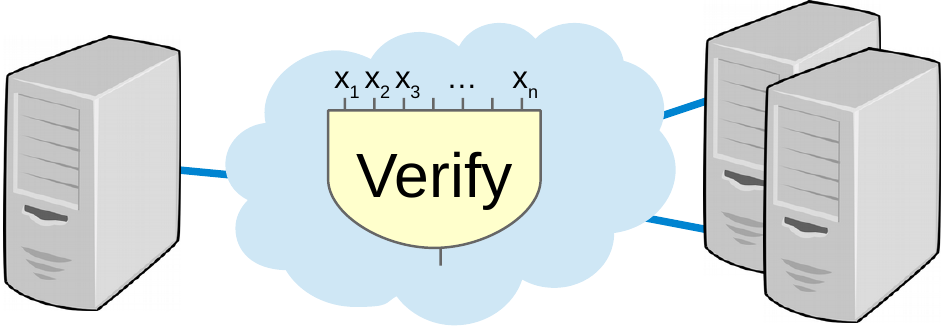}}~~\rulesep{}~\subfloat[If the checks pass, the servers update their local accumulators with the client-provided data.]{\includegraphics[width=\imgwidth]{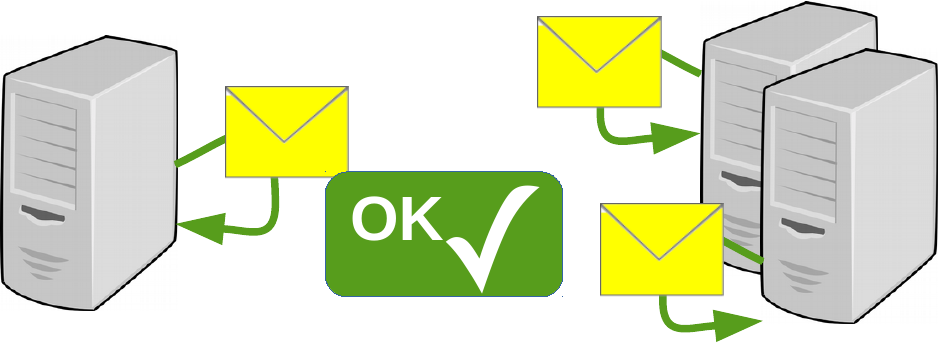}}~~\rulesepThick{}~\subfloat[After accumulating many packets, the servers publish their accumulators to reveal the aggregate.]{\includegraphics[width=\imgwidth]{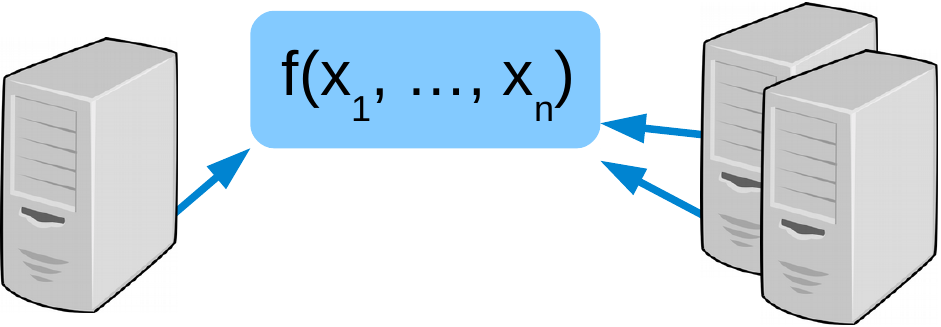}}
\caption{An overview of the \name pipeline for processing client submissions.} 
\label{fig:overview}
\end{figure*}

Let us introduce \name by first presenting a simplified version of it.
In this simple version, 
each client holds a one-bit integer $x_i$ and the servers want to compute
the sum of the clients' private values~$\sum_i x_i$.
Even this very basic functionality has many real-world applications. 
For example, the developer of a health data mobile app
could use this scheme to collect the number of 
app users who have a certain medical condition.
In this application, the bit $x_i$ would indicate 
whether the user has the condition, and the sum over
the $x_i$s gives the count of affected users.

The public parameters for the \name deployment include a prime $p$.
Throughout this paper, when we write 
``$c = a + b \in \F_p$,'' we mean ``$c = a + b \pmod p$.''
The simplified \name scheme for computing sums proceeds in three steps:
\begin{enumerate}

\item \textbf{Upload.}
      Each client $i$ splits its
      private value $x_i$ into $s$ shares, one per server, using a secret-sharing scheme.
      In particular, the client picks random integers $[x_i]_1, \dots, [x_i]_s \in \F_p$,
      subject to the constraint:
      $x_i = [x_i]_1 + \cdots + [x_i]_s \in \F_p$.
      The client then sends, over an encrypted and authenticated channel, one share 
      of its submission to each server.

\item \textbf{Aggregate.}
      Each server $j$ holds an accumulator value $A_j \in \F_p$, initialized to zero.
      Upon receiving a share from the $i$th client, the server adds the uploaded
      share into its accumulator: $A_j \gets A_j + [x_i]_j \in \F_p$.

\item \textbf{Publish.}
      Once the servers have received a share from each client, they publish
      their accumulator values. Computing the sum of the accumulator
      values $\sum_j A_j \in \F_p$ yields the desired sum
      $\sum_i x_i$ of the clients' private values, 
      as long as the modulus $p$ is larger than the number
      of clients 
      (i.e., the sum $\sum_i x_i$ does not ``overflow'' the modulus).
\end{enumerate}

There are two observations we can make about this scheme.
First, even this simple scheme provides privacy: the servers learn 
the sum $\sum_i x_i$ but they learn nothing else about the clients' 
private inputs. 
Second, the scheme does {\em not} provide robustness. 
A single malicious client can completely corrupt the protocol output
by submitting (for example), a random integer $r \in \F_p$ to each server.

The core contributions of \name are to improve this basic scheme in
terms of security and functionality.
In terms of security, \name extends the simple scheme to provide robustness
in the face of malicious clients.
In terms of functionality, \name extends the simple scheme to allow privacy-preserving 
computation of a wide array of aggregation functions (not just sum).

 \section{Protecting correctness with SNIPs}
\label{sec:disrupt}

Upon receiving shares of a client's data value, the \name servers need a way
to check if the client-submitted value is well formed.
For example, in the simplified protocol of Section~\ref{sec:simple},  
every client is supposed to send the servers the share of a value $x$ such
that $0 \leq x \leq 1$.
However, since the client sends only a {\em single share} of its value $x$ to
each server---to preserve privacy---each server essentially receives an encrypted
version of $x$ and cannot unilaterally determine if $x$ is well formed. 
In the more general setting, each \name client submits to each server a share
$[x]_i$ of a vector $x \in \F^L$, for some finite field~$\F$.
The servers hold a validation predicate $\Ver(\cdot)$, and should
only accept the client's data submission if
$\Ver(x) = 1$ (Figure~\ref{fig:overview}).

To execute this check in \name, we introduce a new cryptographic tool
called {\em secret-shared non-interactive proofs} (``SNIPs'').
With these proofs, the client can quickly prove to the servers that
$\Ver(x) = 1$, for an arbitrary function $\Ver$, without leaking anything else
about~$x$ to the servers.

\paragraph{Building blocks.}
All arithmetic in this section takes place in a finite field $\F$,
or modulo a prime $p$, if you prefer.
We use a simple additive secret-sharing scheme over $\F$:
to split a value $x \in \F$ into $s$ shares, choose
random values $([x]_1, \dots, [x]_s) \in \F^s$
subject to the constraint that $x = \sum_i[x]_i \in \F$.
In our notation, $[x]_i$ denotes the $i$th share of $x$.
An adversary who gets hold of any subset of up to $s-1$
shares of $x$ learns nothing, in an information-theoretic sense,
about $x$ from the shares.

This secret-sharing scheme is linear, which means that the servers can
perform affine operations on shares without communicating.
That is, by adding shares $[x]_i$ and $[y]_i$, 
the servers can locally construct shares $[x+y]_i$. 
Given a share $[x]_i$, the servers can also construct 
shares $[\alpha x + \beta]_i$, for any constants $\alpha,\beta \in \F$.
(This is a classic observation from the multi-party computation 
literature~\cite{ben1988completeness}.)

Our construction uses {\em arithmetic circuits}.
An arithmetic circuit is like a boolean circuit except that it uses
finite-field multiplication, addition, and multiplication-by-constant gates,
instead of boolean \textsc{and}, \textsc{or}, and \textsc{not} gates.
See Appendix~\ref{app:beaver:circuit} for a formal definition.

\subsection{Overview} 
\label{sec:disrupt:overview}

A secret-shared non-interactive proof (SNIP) protocol consists
of an interaction between a client (the prover)
and multiple servers (the verifiers).
At the start of the protocol:
\begin{itemize}
\item[--]
  each server $i$ holds a vector $[x]_i \in \F^L$, 
\item[--]
  the client holds the vector $x = \sum_i [x]_i \in \F^L$, and
\item[--]
all parties hold an arithmetic circuit representing a predicate
$\Ver : \F^L \to \F$.
\end{itemize}
The client's goal is to convince the servers that $\Ver(x)=1$,
without leaking anything else about $x$ to the servers.
To do so, the client sends a proof string to each server.
After receiving these proof strings, the servers 
gossip amongst themselves
and then conclude either that $\Ver(x)=1$ 
(the servers ``accept $x$'') or not
(the servers ``reject $x$'').

For a SNIP to be useful in \name, 
it must satisfy the following properties:

\begin{description}
\item[Correctness.]
    If all parties are honest, the servers will accept $x$.

\item[Soundness.]
    If all servers are honest, and if $\Ver(x) \neq 1$, 
    then for all malicious clients, even ones running in super-polynomial
    time, 
    the servers will reject $x$ with overwhelming probability.
    In other words, no matter how the client cheats, the servers
    will almost always reject $x$. 

\item[Zero knowledge.]
    If the client and at least one server are honest, 
    then the servers learn nothing about $x$, 
    except that $\Ver(x) = 1$.
    More precisely, there exists a simulator (that does not take $x$
    as input) that accurately reproduces the view of any 
    proper subset of malicious servers executing the SNIP protocol.
\end{description}
See Appendix~\ref{app:zkproof} for formal definitions.

\noindent
These three security properties are nearly identical 
to the properties required of a zero-knowledge interactive 
proof system~\cite{goldwasser1989knowledge}.
However, in the conventional zero-knowledge setting, 
there is a single prover and single verifier,
whereas here we have a single prover (the client) and 
{\em many} verifiers (the servers).

\itpara{Our contribution.}
We devise a SNIP that requires minimal server-to-server
communication, is compatible with any public $\Ver$ circuit, and relies solely on
fast, information-theoretic primitives.
(We discuss how to hide the $\Ver$ circuit from the client in
Section~\ref{sec:disrupt:serv}.)

To build the SNIP, we first generalize a ``batch verification'' technique of 
Ben-Sasson et al.~\cite{ben2012near} and then show how a set of servers
can use it to verify an entire circuit computation by exchanging
a only few field elements.
We implement this last step with a new adaptation of Beaver's multi-party
computation (MPC) protocol to the client/server setting~\cite{beaver1991efficient}.

\itpara{Related techniques.}
Prior work has studied interactive proofs in both 
the many-prover~\cite{ben1988multi,fortnow1994power}
and many-verifier settings~\cite{beaver1991secure,abe2002non}.
Prior many-verifier protocols require relatively expensive
public-key primitives~\cite{abe2002non} or require an amount of
server-to-server traffic that grows linearly in the size of the 
circuit for the $\Ver$ function~\cite{beaver1991secure}.
In concurrent independent work, Boyle et al.~\cite{boyle2016function}
construct what we can view as a very efficient SNIP for a {\em specific} 
$\Ver$ function, in which the servers are semi-honest 
(``honest but curious'')~\cite{boyle2016function}.
They also use a Beaver-style MPC multiplication;
their techniques otherwise differ from ours.

\subsection{Constructing SNIPs}
\label{sec:disrupt:snip}

To run the SNIP protocol,
the client and servers execute the following steps:

\itpara{Set-up.}
Let $M$ be the number of multiplication gates in 
the arithmetic circuit for $\Ver$.
We work over a field $\F$ that is large enough to
ensure that $2M+2 \ll |\F|$.

\newcounter{SnipCounter}

\newcommand{\stepCount}[1]{\refstepcounter{SnipCounter}\label{#1}\arabic{SnipCounter}}
\paragraph{Step~\stepCount{proto:eval}: Client evaluation.}
The client evaluates the $\Ver$ circuit on its input $x$.
The client thus knows the value that every wire in the 
circuit takes on during the computation of $\Ver(x)$.
The client uses these wire values to construct
three randomized polynomials $f$, $g$, and $h$, which
encode the values on the input and output wires of each of
the $M$ multiplication gates in the $\Ver(x)$ computation. 

Label the $M$ multiplication gates in the $\Ver(x)$ circuit,
in topological order from inputs to outputs,
with the numbers $\{1, \dots, M\}$.
For $t \in \{1, \dots, M\}$, let us define $u_t$ and $v_t$ to be the
values on the left and right input wires of the $t$-th multiplication gate.
The client chooses $u_0$ and $v_0$ to be values
sampled independently and uniformly at random from $\F$.
Then, define $f$ and $g$ to be the lowest-degree
polynomials such that $f(t) = u_t$ and $g(t) = v_t$,
for all $t \in \{0,\dots,M\}$.
Finally, define the polynomial $h$ as~${h = f \cdot g}$.

The polynomials $f$ and $g$ will have degree at most $M$, and the
polynomial $h$ will have degree at most $2M$.
Since $h(t) = f(t)\cdot g(t) = u_t \cdot v_t$ for all $t \in \{1, \dots, M\}$, 
$h(t)$ is equal to the value of the output wire ($u_t \cdot v_t$) 
of the $t$-th multiplication gate in the $\Ver(x)$ circuit,
for $1 \leq t \leq M$.

In Step~\ref{proto:eval} of the checking protocol, 
the client executes the computation of $\Ver(x)$, 
uses polynomial interpolation to construct the
polynomials $f$ and $g$, and multiplies these polynomials
to produce $h = f \cdot g$.
The client then splits the random values $f(0) = u_0$ and $g(0) = v_0$,
using additive secret sharing, 
and send shares $[f(0)]_i$ and $[g(0)]_i$ to server $i$.
The client also splits the coefficients of $h$ 
and sends the $i$th share of the coefficients $[h]_i$ to server $i$.

\paragraph{Step~\stepCount{proto:poly}: Consistency checking at the servers.} 
Each server $i$ holds a share $[x]_i$ of the client's private value~$x$.
Each server also holds shares $[f(0)]_i$, $[g(0)]_i$, and $[h]_i$.
Using these values, each server can---without communicating 
with the other servers---produce shares $[f]_i$ and $[g]_i$ 
of the polynomials $f$ and $g$.

To see how, first observe that if a server has a share of every wire
value in the circuit, along with shares of $f(0)$ and $g(0)$, 
it can construct $[f]_i$ and $[g]_i$ using polynomial interpolation.
Next, realize that each server can reconstruct 
a share of every wire value in the circuit 
since each server:
\begin{itemize}
  \item has a share of each of the input wire values~($[x]_i$), 
  \item has a share of each wire value coming out of a multiplication gate
        (for $t \in \{1, \dots, M\}$, the value $[h]_i(t)$ is a share of the $t$-th such wire), and
  \item can derive all other wire value shares via affine operations
        on the wire value shares it already has.
\end{itemize}
Using these wire value shares, the servers use polynomial interpolation
to construct $[f]_i$ and $[g]_i$.

If the client and servers have acted honestly up to this point, 
then the servers will now hold shares of polynomials 
$f$, $g$, and $h$ such that $f \cdot g = h$.

In contrast, a malicious client could have sent the servers 
shares of a polynomial $\hat{h}$ such that, for some $t \in \{1, \dots, M\}$,
$\hat{h}(t)$ is {\em not} the value on the output wire in the $t$-th
multiplication gate of the $\Ver(x)$ computation.
In this case, the servers will reconstruct shares of polynomials 
$\hat{f}$ and $\hat{g}$ that might not be equal to $f$ and $g$.
We will then have with certainty 
that $\hat{h} \neq \hat{f} \cdot \hat{g}$.
To see why, consider the least $t_0$ for which $\hat{h}(t_0) \neq h(t_0)$.
For all $t \leq t_0$, $\hat{f}(t) = f(t)$ and $\hat{g}(t) = g(t)$, by
construction.
Since 
\[ \hat{h}(t_0) \neq h(t_0) = f(t_0)\cdot g(t_0) = \hat{f}(t_0) \cdot \hat{g}(t_0),\]
it must be that $\hat{h}(t_0) \neq  \hat{f}(t_0) \cdot \hat{g}(t_0)$, so $\hat{h} \neq \hat{f} \cdot \hat{g}$.
(Ben-Sasson et al.~\cite{ben2012near} use polynomial identities
to check the consistency of secret-shared values in 
a very different MPC protocol.
Their construction inspired our approach.)

\paragraph{Step~\stepCount{proto:test}a: Polynomial identity test.}
At the start of this step, each server $i$ holds shares 
$[\hatf]_i$, $[\hatg]_i$, and $[\hath]_i$ of polynomials 
$\hatf$, $\hatg$, and $\hath$.
Furthermore, it holds that $\hatf \cdot \hatg = \hath$
if and only if the servers collectively hold a set of wire value shares that,
when summed up, equal the internal wire values of the $\Ver(x)$ circuit computation.

The servers now execute a variant of the Schwartz-Zippel randomized polynomial
identity test~\cite{schwartz1980fast,zippel1979probabilistic} to check whether
this relation holds.
The principle of the test is that if $\hatf(t) \cdot \hatg(t) \neq \hath(t)$,
then the polynomial $t \cdot (\hatf(t) \cdot \hatg(t) - \hath(t))$ 
is a non-zero polynomial of degree at most $2M + 1$.
(Multiplying the polynomial $\hatf \cdot g - \hath$ by $t$ 
is useful for the next step.)
Such a polynomial can have at most $2M+1$ zeros in $\F$, so if we choose
a random $r \in \F$ and evaluate 
$r \cdot (\hatf(r) \cdot \hatg(r) - \hath(r))$, the servers will detect that
$\hatf \cdot \hatg \neq \hath$ with probability at least $1 - \frac{2M+1}{|\F|}$.

To execute the polynomial identity test, one of the servers samples a random value $r \in \F$.
Each server $i$ then evaluates her share of each of the three polynomials
on the point $r$ to get $[\hatf(r)]_i$, $[\hatg(r)]_i$, and $[\hath(r)]_i$.
The servers can perform this step locally, since polynomial evaluation
requires only affine operations on shares.
Each server then applies a local linear operation to these last
two shares to produce shares of $[r \cdot \hatg(r)]_i$, and ${[r \cdot \hath(r)]_i}$.

Assume for a moment that each server $i$ can multiply her shares
$[\hatf(r)]_i$ and $[r \cdot \hatg(r)]_i$ to produce a 
share ${[r \cdot \hatf(r) \cdot \hatg(r)]_i}$.
In this case, the servers can use a linear operation to get
shares $\sigma_i = [r\cdot(\hatf(r) \cdot \hatg(r) - \hath(r))]_i$. 
The servers then publish these $\sigma_i$s and 
ensure that ${\sum_i \sigma_i = 0 \in \F}$.
The servers reject the client's submission if $\sum_i \sigma_i \neq 0$.

\paragraph{Step~\ref{proto:test}b: Multiplication of shares.}
Finally, the servers must somehow multiply their shares $[\hatf(r)]_i$
and $[r\cdot \hatg(r)]_i$ to get a share $[r \cdot \hatf(r) \cdot \hatg(r)]_i$ without
leaking anything to each other about the values $\hatf(r)$ and $\hatg(r)$.
To do so, we adapt a multi-party computation (MPC) technique of Beaver~\cite{beaver1991efficient}.
The details of Beaver's MPC protocol are not critical here, but
we include them for reference in Appendix~\ref{app:beaver:proto}.

Beaver's result implies that if servers receive, from a trusted
dealer, one-time-use shares $([a]_i, [b]_i, [c]_i) \in \F^3$ of random values 
such that $a \cdot b = c \in \F$
(``multiplication triples''), then the servers can very efficiently
execute a multi-party multiplication of a pair secret-shared values. 
Furthermore, the multiplication protocol is fast: 
it requires each server to broadcast a single message.

In the traditional MPC setting, the parties to the computation have to run an expensive
cryptographic protocol to generate the multiplication triples themselves~\cite{SPDZ}.
In our setting however, the client generates the multiplication triple
on behalf of the servers:
the client chooses $(a, b, c) \in \F^3$ such that $a\cdot b = c \in \F$,
and sends shares of these values to each server.
If the client produces shares of these values correctly, then the servers
can perform a multi-party multiplication of shares to complete 
the correctness check of the prior section.

Crucially, we can ensure that {\em even if the client sends shares of an invalid
multiplication triple to the servers}, 
the servers will still catch the cheating client with high probability.
First, say that a cheating client sends the servers shares 
$([a]_i, [b]_i, [c]_i) \in \F^3$ such that $a \cdot b \neq c \in \F$.
Then we can write $a \cdot b = (c + \alpha) \in \F$, 
for some constant $\alpha > 0$.

In this case, when the servers run Beaver's MPC multiplication protocol to
execute the polynomial identity test, the result of the test will be shifted
by $\alpha$.
(To confirm this, consult our summary of Beaver's protocol in
Appendix~\ref{app:beaver:proto}.)
So the servers will effectively testing whether the polynomial
\[\hat{P}(t) = t\cdot(\hatf(t) \cdot \hatg(t) - \hath(t)) + \alpha\]
is identically zero.
Whenever $\hatf \cdot \hatg \neq \hath$, 
it holds that
$t \cdot (\hatf(t) \cdot \hatg(t) - \hath(t))$ is a non-zero polynomial.
So, if ${\hatf \cdot \hatg \neq \hath}$, 
then $\hat{P}(t)$ must also be a non-zero polynomial.
(In constructing $\hat{P}$, We multiply the term ``${\hatf \cdot \hatg - \hath}$'' by $t$, to 
ensure that whenever this expression is non-zero, the resulting polynomial $\hat{P}$ is 
non-zero, even if $\hatf$, $\hatg$, and $\hath$ are degree-zero polynomials, and the
client chooses $\alpha$ adversarially.)

Since we only require soundness to hold if all servers are honest,
we may assume that
the client did not know the servers' random value $r$ when 
the client generated its multiplication triple.
This implies that $r$ is distributed independently of $\alpha$, 
and since we only require soundness to hold if the servers are honest,
we may assume that $r$ is sampled uniformly from $\F$ as well.

So, even if the client cheats, the servers will still be executing the polynomial
identity test on a non-zero polynomial of degree at most $2M+1$. 
The servers will thus catch a cheating client with probability at least $1 -
\frac{2M+1}{|\F|}$.
In Appendix~\ref{app:zkproof:sound}, we present a formal definition of the soundness
property and we prove that it holds.

\paragraph{Step~\stepCount{proto:output}: Output verification.} 
If all servers are honest, at the start of 
the final step of the protocol, each server $i$ will
hold a set of shares of the values that the $\Ver$
circuit takes on during computation of $\Ver(x)$:
$([w_1]_i, [w_2]_i, \dots)$.
The servers already hold shares of the input wires of this
circuit ($[x]_i$), so to confirm that $\Ver(x)=1$, the
servers need only publish their shares of the output wire.
When they do, the servers can sum up these shares to confirm
that the value on the output wire is equal to one, in which
case it must be that $\Ver(x)=1$, except with 
some small failure probability due to the polynomial identity test.

\medskip

Combining all of the pieces above: 
each share of a SNIP proof consists of
a share of the client-produced tuple ${\pi = (f(0), g(0), h, a, b, c)}$.

\subsection{Security and efficiency}
The correctness of the scheme follows by construction.
To trick the servers into accepting a malformed submission, 
a cheating client must subvert the polynomial identity test.
This bad event has probability at most
$(2M+1)/|\F|$, where $M$ is the number of multiplication gates
in $\Ver(\cdot)$.
By taking $|\F| \approx 2^{128}$, or repeating Step~\ref{proto:test} a few
times, we can make this failure probability extremely small.

We require neither completeness nor soundness to hold 
in the presence of malicious {\em servers},
though we do require soundness against malicious clients.
A malicious server can thus trick the honest servers into rejecting a
well-formed client submission that they should have accepted.
This is tantamount to the malicious server mounting a selective denial-of-service
attack against the honest client.
We discuss this attack in Section~\ref{sec:disc:attacks}.

We prove in Appendix~\ref{app:zkproof:zero} that, as long as there is at least
one honest server, the dishonest servers gain no information---in an unconditional, information-theoretic sense---about
the client's data values nor about the values on the internal wires in the
$\Ver(x)$ circuit.

\newcommand*\rotbox{\rotatebox{90}}
\definecolor{lightyellow}{HTML}{FFFF99}
\newcommand*\ycell{\cellcolor{lightyellow}}
\begin{table}
\centering
  {\footnotesize
  \begin{tabular}{l l | c c | c}
  &\multicolumn{1}{c}{}   & NIZK & \multicolumn{1}{c}{SNARK} & {\bf \Name} (SNIP) \\ \hline
    \multirow{3}{*}{\rotbox{\bf Client}}& \ycell Exps.& \ycell $M$ & \ycell $M$ & \ycell $0$\\
                                        &Muls.       & $0$ & $M\log M$& $M \log M$\\
                                        &Proof len. & $M$ & $1$ & $M$ \\ \hline
    \multirow{3}{*}{\rotbox{\bf Servers}}& \ycell Exps./Pairs. & \ycell $M$ & \ycell $1$ & \ycell $0$ \\
                                        &Muls.       & $0$ & $M$ & $M \log M$\\
                                        &Data transfer& $M$ & $1$ & $1$ \\ \hline
  \end{tabular}
  }
  \caption{An asymptotic comparison of \name with standard zero-knowledge techniques showing
    that \name reduces the computational burden for clients and servers.
    The client holds a vector $x \in \F^M$, each server $i$ holds a share $[x]_i$, and
  the client convinces the servers that each component of $x$ is a $0/1$ value in $\F$.
  We suppress the $\Theta(\cdot)$ notation for readability.}
  \label{tab:efficiency}
\end{table}

\paragraph{Efficiency.}
The remarkable property of this SNIP construction is that the server-to-server
communication cost grows neither with the complexity of the verification circuit nor
with the size of the value $x$ (Table~\ref{tab:efficiency}).
The computation cost at the servers is essentially the same as the cost for
each server to evaluate the $\Ver$ circuit locally.
That said, the client-to-server communication cost does grow linearly with 
the size of the $\Ver$ circuit.
An interesting challenge would be to try to reduce the client's bandwidth usage without
resorting to relatively expensive public-key cryptographic techniques~\cite{parno2013pinocchio,ben2014succinct,ben2014scalable,costello2015geppetto,braun2013verifying}.

\subsection{Computation at the servers}
\label{sec:disrupt:serv}

Constructing the SNIP proof requires the client to compute $\Ver(x)$ on its own. 
If the verification circuit takes secret server-provided values as input, or
is itself a secret belonging to the servers,
then the client does not have enough information to compute $\Ver(x)$. 
For example, the servers could run a proprietary verification algorithm to detect
spammy client submissions---the servers would want to run this algorithm without
revealing it to the (possibly spam-producing) clients.
To handle this use case, 
the servers can execute the verification check themselves
at a slightly higher cost.  

This variant maintains privacy only against ``honest but curious''
servers---in contrast, the SNIP-based variant maintains privacy
against actively malicious servers.
See Appendix~\ref{app:serverside} for details.

 \section{Gathering complex statistics}
\label{sec:struct}

So far, we have developed the means to compute private sums
over client-provided data
(Section~\ref{sec:simple}) and to check
an arbitrary validation predicate against secret-shared data 
(Section~\ref{sec:disrupt}).
Combining these two ideas with careful data encodings, which
we introduce now, allows \name to compute more 
sophisticated statistics over private client data.
At a high level, each client first encodes its private data value
in a prescribed way, and the servers then privately compute the sum of the encodings.
Finally, the servers can decode the summed encodings to recover the statistic of interest.
The participants perform this encoding and decoding via a mechanism 
we call affine-aggregatable encodings (``AFEs'').

\subsection{Affine-aggregatable encodings (AFEs)}
In our setting, each client $i$ holds a value $x_i\in \D$,
where $\D$ is some set of data values.
The servers hold an aggregation function $f: \D^n \to \A$, whose range is a
set of aggregates $\A$.
For example, the function $f$ might compute the standard deviation of 
its $n$ inputs.
The servers'  goal is to evaluate $f(x_1,\ldots,x_n)$ without
learning the $x_i$s.

An \afe gives an efficient way to encode the
data values $x_i$ such that it is possible
to compute the value $f(x_1, \dots, x_n)$ 
given only the {\em sum of the encodings} of $x_1,\ldots,x_n$. 
An \afe consists of three efficient algorithms $(\Enc, \Ver, \Dec)$, defined
with respect to a field~$\F$ and two integers $k$ and $k'$,
where $k' \leq k$:
\begin{itemize}[leftmargin=\parindent]
\item $\Enc(x)$: maps an input $x \in \D$
      to its encoding in $\F^k$,
\item $\Ver(y)$: returns true if and only if $y \in \F^k$ is a valid
      encoding of some data item in $\D$,
\item $\Dec(\sigma)$: 
      takes $\sigma = \sum_{i=1}^n \proj{k'}\big(\Enc(x_i)\big) \in \F^{k'}$ 
      as input,
      and outputs $f(x_1, \dots, x_n)$. 
      The $\proj{k'}(\cdot)$ function
      outputs the first $k' \leq k$ components of its input.
\end{itemize}
The \afe uses all $k$ components of the encoding in validation,
but only uses $k'$ components to decode~$\sigma$.
In many of our applications we have $k' = k$.  

An \afe is {\em private with respect to a function $\hat{f}$},
or simply $\hat{f}$-private, if 
$\sigma$ reveals nothing about $x_1,\ldots,x_n$ beyond
what $\hat{f}(x_1, \dots, x_n)$ itself reveals. 
More precisely, it is possible to efficiently simulate $\sigma$ 
given only $\hat{f}(x_1, \dots, x_n)$.  
Usually $\hat{f}$ reveals nothing more than the aggregation function $f$ 
(i.e., the minimum leakage possible), but in some
cases $\hat{f}$ reveals a little more than~$f$. 

For some functions $f$ we can build more efficient $f$-private \afes
by allowing the encoding algorithm to be randomized.  In these cases,
we allow the decoding algorithm to return an answer that is only an
approximation of $f$, and we also allow it to fail with some
small probability.  

Prior systems have made use of specific \afes for 
sums~\cite{elahi2014privex,kursawe2011},
standard deviations~\cite{popa2011privacy}, 
counts~\cite{melis2016,broadbent2007information}, and
least-squares regression~\cite{karr2005secure}.
Our contribution is to unify these notions
and to adopt existing \afes to enable better
composition with \name's SNIPs.
In particular, by using more complex encodings, we can 
reduce the size of the circuit for $\Ver$, which results in
shorter SNIP proofs.

\paragraph{\afes in \name: Putting it all together.}
The full \name system
computes $f(x_1,\ldots,x_n)$ privately as follows (see Figure~\ref{fig:overview}):
Each client encodes its data value $x$ using the \afe $\Enc$ routine
for the aggregation function $f$.
Then, as in the simple scheme of Section~\ref{sec:simple},
every client splits its encoding into $s$ shares and sends one share
to each of the $s$ servers.
The client uses a SNIP proof (Section~\ref{sec:disrupt}) 
to convince the servers that its encoding satisfies the \afe $\Ver$ predicate.

Upon receiving a client's submission, the servers verify the 
SNIP to ensure that the encoding is well-formed.  
If the servers conclude that the encoding is valid, every server adds the
first $k'$ components of the encoding share to its local running accumulator.
(Recall that $k'$ is a parameter of the \afe scheme.)
Finally, after collecting valid submissions from many clients,
every server publishes its local accumulator, enabling
anyone to run the \afe $\Dec$ routine to 
compute the final statistic in the clear.
The formal description of the system is presented in 
Appendix~\ref{app:security}, where we also analyze its security. 

\itpara{Limitations.}
There exist aggregation functions for which 
all \afe constructions must have large encodings.
For instance, say that each of $n$ clients holds an integer
$x_i$, where $1 \leq x_i \leq n$. 
We might like an \afe that computes the median of these integers
$\{x_1, \dots, x_n\}$, working over a field $\F$ with $|\F| \approx n^d$,
for some constant $d \geq 1$.

We show that there is no such AFE whose encodings consist of 
$k' \in o(n/\log n)$ field elements.
Suppose, towards a contradiction, that such an AFE did exist.
Then we could describe any sum of encodings using
at most $O(k' \log |\F|) = o(n)$ bits of information.
From this \afe, we could build a single-pass, space-$o(n)$
streaming algorithm for computing the exact median of an $n$-item stream. 
But every single-pass streaming algorithm for
computing the exact median over an $n$-item stream requires $\Omega(n)$ 
bits of space~\cite{guha2009stream}, which is a contradiction.
Similar arguments may rule out space-efficient \afe constructions 
for other natural functions.

\subsection{Aggregating basic data types}
\label{sec:struct:basic}

This section presents the basic affine-aggregatable encoding schemes 
that serve as building blocks for the more sophisticated schemes. 
In the following constructions, the clients hold
data values $x_1, \dots, x_n \in \D$, and our goal
is to compute an aggregate $f(x_1, \dots, x_n)$.

In constructing these encodings, we have two goals.
The first is to ensure that the \afe leaks as little as possible
about the $x_i$s, apart from the value $f(x_1, \dots, x_n)$ itself.
The second is to minimize the number of multiplication gates in the
arithmetic circuit for $\Ver$, since the cost of the SNIPs
grows with this quantity.

In what follows, we let $\lambda$ be a security
parameter, such as $\lambda = 80$ or $\lambda = 128$.

\paragraph{Integer sum and mean.}
We first construct an \afe{} for computing the sum of $b$-bit integers. 
Let $\F$ be a finite field of size at least $n 2^b$. 
On input $0 \leq x \leq 2^b - 1$, 
the $\Enc(x)$ algorithm first computes the bit representation of $x$,
denoted $(\beta_0, \beta_1, \dots, \beta_{b-1}) \in \{0,1\}^b$.
It then treats the binary digits as elements of $\F$, and outputs
$$\Enc(x) = (x, \beta_0, \ldots, \beta_{b-1}) \in \F^{b+1}.$$

To check that $x$ represents a $b$-bit integer, the $\Ver$
algorithm ensures that
each $\beta_i$ is a bit, and that the bits 
represent $x$.
Specifically, the algorithm checks that the following equalities hold over $\F$:
\[  \Ver(\Enc(x)) = \bigg(x = \sum_{i=0}^{b-1} 2^i \beta_i\bigg)
     \land \bigwedge_{i=1}^n \bigg[( \beta_i - 1) \cdot \beta_i = 0 \bigg]. \]

The $\Dec$ algorithm takes the sum of encodings $\sigma$ as input, 
truncated to only the first coordinate.
That is, 
$\sigma = \sum_{i=1}^n \proj{1}\big(\Enc(x_1)\big) = x_1 + \cdots +x_n$.
This $\sigma$ is the required aggregate output. 
Moreover, this \afe is clearly sum-private.

To compute the arithmetic mean, we 
divide the sum of integers by $n$ over the rationals. 
Computing the product and geometric mean works in exactly 
the same matter, except that we encode $x$ using $b$-bit logarithms.

\paragraph{Variance and \textsc{stddev}.}
Using known techniques~\cite{castelluccia2005efficient,popa2011privacy},
the summation AFE above lets us compute the variance of a set of $b$-bit 
integers using the identity:
$\text{Var}(X) = \text{E}[X^2] - (\text{E}[X])^2$.
Each client encodes its integer $x$ as $(x, x^2)$ and then applies
the summation AFE to each of the two components.
(The $\Ver$ algorithm also ensures that second
integer is the square of the first.)
The resulting two values let us compute the variance.

This \afe also reveals the expectation $\text{E}[X]$.
It is private with respect to the function $\hat{f}$ that outputs
both the expectation and variance of the given set of integers. 

\paragraph{Boolean \textsc{or} and \textsc{and}.}
When $\D = \{0,1\}$ and $f(x_1,\ldots,x_n) = \textsc{or}(x_1,\ldots,x_n)$
the encoding operation outputs an element of $\F_2^\lambda$ 
(i.e., a $\lambda$-bit bitstring) as:
\[ \Enc(x) = \begin{cases}
   \text{$\lambda$ zeros}               & \text{if $x = 0$} \\
   \text{a random element in $\F_2^\lambda$} & \text{if $x = 1$.} 
\end{cases}  \]
The $\Ver$ algorithm outputs ``1'' always, since all $\lambda$-bit encodings
are valid.  
The sum of encodings is simply the \textsc{xor}
of the $n$ $\lambda$-bit encodings. 
The $\Dec$ algorithm takes as input a $\lambda$-bit string and outputs
``0'' if and only if its input is a $\lambda$-bit string of zeros.
With probability $1 - 2^{-\lambda}$, 
over the randomness of the encoding algorithm,
the decoding operation returns the boolean \textsc{or} of the encoded values.
This \afe is \textsc{or}-private. 
A similar construction yields an \afe for boolean \textsc{and}.

\paragraph{\textsc{min} and \textsc{max}.}
To compute the minimum and maximum of integers over a range
$\{0, \dots, B-1\}$, where $B$ is small
(e.g., car speeds in the range $0$--$250$ km/h), the
$\Enc$ algorithm can represent each integer in unary as a length-$B$ vector
of bits $(\beta_0, \dots, \beta_{B-1})$, where $\beta_i = 1$ if
and only if the client's value $x \leq i$.
We can use the bitwise-\textsc{or} construction above to 
take the \textsc{or} of the client-provided vectors---the largest
value containing a ``1'' is the maximum.
To compute the minimum instead, replace \textsc{or} with \textsc{and}.
This is \textsc{min}-private, as in the \textsc{or} protocol above. 

When the domain is large (e.g., we want the \textsc{max} 
of 64-bit packet counters, in a networking application), we can 
get a $c$-approximation of the \textsc{min} and \textsc{max} using 
a similar idea: 
divide the range $\{0, \dots, B-1\}$ into $b = \log_c B$ ``bins''
$[0, c), [c, c^2), \dots, [c^{b-1}, B)$.
Then, use the small-range \textsc{min}/\textsc{max} construction,
over the $b$ bins, 
to compute the approximate statistic.
The output will be within a multiplicative factor of $c$ of
the true value.  
This construction is private with respect to the approximate
\textsc{min}/\textsc{max} function. 

\paragraph{Frequency count.}
Here, every client has a value $x$ in a small set of data values $\D = \{0, \dots, B-1\}$.
The goal is to output a $B$-element vector $v$, where $v[i]$ is the number
of clients that hold the value $i$, for every $0 \leq i < B$.  

Let $\F$ be a field of size at least $n$.
The $\Enc$ algorithm encodes a value $x \in \D$ as a length-$B$ 
vector $(\beta_0, \dots, \beta_{B-1}) \in \F^B$
where $\beta_i = 1$ if $x = i$ and $\beta_i = 0$ otherwise.
The $\Ver$ algorithm checks that each $\beta$ value is in the set $\{0,1\}$
and that the sum of the $\beta$s is exactly one. 
The $\Dec$ algorithm does nothing: the final output
is a length-$B$ vector, whose $i$th component gives the number of clients
who took on value $i$.  Again, this \afe is private with respect to the
function being computed. 

The output of this \afe yields enough information to compute
other useful functions (e.g., quantiles) of the
distribution of the clients' $x$ values.
When the domain $\D$ is large, this \afe is very inefficient.
In Appendix~\ref{app:afe-extra}, we give \afes for 
approximate counts over large domains.

\paragraph{Sets.}
We can compute the intersection or union of sets over a small universe
of elements using the boolean \afe operations: represent
a set of $B$ items as its characteristic vector of booleans, 
and compute an \textsc{and}
for intersection and 
an \textsc{or} for union. 
When the universe is large, the approximate \afes
of Appendix~\ref{app:afe-extra} are more efficient.

\subsection{Machine learning}
\label{sec:struct:ml}
We can use \name for training machine learning models on private client data.
To do so, we exploit the observation of Karr et al.~\cite{karr2005secure} that a
system for computing private sums can also privately train linear models.
(In Appendix~\ref{app:afe-extra}, we also show how to use \name to privately 
evaluate the $R^2$-coefficient of an existing model.)
In \name, we extend their work by showing how to perform these
tasks while maintaining robustness against malicious clients.

Suppose that every client holds a data point $(x,y)$ where 
$x$ and $y$ are $b$-bit integers. 
We would like to train a model that takes $x$ as input and outputs
a real-valued prediction $\hat{y}_i = M(x) \in \RR$ of $y$.
We might predict a person's blood pressure ($y$)
from the number of steps they walk daily ($x$).

\begin{table}
  \centering
  {\small
  \begin{tabular}{r r r r r r r}\\
    && \multicolumn{2}{c}{\textbf{Workstation}} & \multicolumn{2}{c}{\textbf{Phone}}\\
    \multicolumn{2}{r}{\em Field size:}& $87$-bit & $265$-bit & $87$-bit & $265$-bit \\ \hline
    \multicolumn{2}{r}{Mul.~in field ($\mu{}s$)}  & $1.013$ & $1.485$ & $11.218$ & $14.930$\\ \hline
        \multirow{3}{0.5cm}{\fontsize{9pt}{9pt}\selectfont{\Name client time}}&$L=10^1$& $0.003$ & $0.004$ & $0.017$ & $0.024$\\
                                      &$L=10^2$& $0.024$ & $0.036$ & $0.112$ & $0.170$\\
                                      &$L=10^3$& $0.221$ & $0.344$ & $1.059$ & $2.165$\\ \hline
                     \end{tabular}
  }
  \caption{Time in seconds for a client to 
    generate a \name submission
    of $L$ four-bit integers to be summed at the servers.
    Averaged over eight runs.}
        \label{tab:bench}
\end{table}

We wish to compute the least-squares linear fit 
$h(x) = c_0 + c_1 x$ over all of the client points.  
With $n$ clients, the model coefficients $c_0$ and $c_1$ satisfy
the linear relation:
\begin{equation} \label{eq:minsq}
  \begin{pmatrix}  
         n                & \sum_{i=1}^n x_i \\[2mm]
         \sum_{i=1}^n x_i  & \sum_{i=1}^n x_i^2
    \end{pmatrix} 
  \cdot
    \begin{pmatrix}  c_0 \\[2mm] c_1  \end{pmatrix}  
  =
    \begin{pmatrix}  \sum_{i=1}^n y_i \\[2mm] \sum_{i=1}^n x_i y_i \end{pmatrix}
\end{equation}

To compute this linear system in an \afe, every client encodes her
private point $(x,y)$ as a vector
\[ (x, x^2, y, x y,\quad  
    \beta_0,\ldots,\beta_{b-1},\quad
    \gamma_0,\ldots,\gamma_{b-1}) \in \F^{2b+4},
\]
where $(\beta_0,\ldots,\beta_{b-1})$ is the binary representation of
$x$ and $(\gamma_0,\ldots,\gamma_{b-1})$ is the binary representation of
$y$.  The validation algorithm checks that all the $\beta$ and $\gamma$
are in $\{0,1\}$, and that all the arithmetic relations hold, analogously
to the validation check for the integer summation \afe.
Finally, the decoding algorithm takes as input
the sum of the encoded vectors truncated to the first four components:
\[  \sigma = \textstyle{\big(\ \sum_{i=1}^n x,\quad  \sum_{i=1}^n x^2,\quad
                    \sum_{i=1}^n y,\quad  \sum_{i=1}^n x y \ \big)  },
\]
from which the decoding algorithm computes the required
real regression coefficients $c_0$ and $c_1$ using~{(\ref{eq:minsq})}.
This \afe is private with respect to the function that outputs the 
least-squares fit $h(x) = c_0 + c_1 x$, along 
with the mean and variance of the set $\{x_1,\ldots,x_n\}$.

When $x$ and $y$ are real numbers, we can embed the reals into a finite
field $\F$ using a fixed-point representation, as long as we size the field large
enough to avoid overflow.

The two-dimensional approach above generalizes directly to 
perform linear regression on $d$-dimensional feature vectors
$\bar{x} = (x^{(1)}, \ldots, x^{(d)})$.
The \afe yields a least-squares approximation of the form
$h(\bar{x}) = c_0 + c_1 x^{(1)} + \cdots + c_d x^{(d)}$.
The resulting \afe is private with respect
to a function that reveals the least-square coefficients $(c_0, \ldots,c_d)$,
along with the $d \times d$ covariance matrix 
$\sum_i { \bar{x}_{i} \cdot (\bar{x}_{i})^T }$.

\section{Evaluation}
\label{sec:eval}

In this section, we demonstrate that
\name's theoretical contributions translate into practical 
performance gains.
We have implemented a \name prototype in 5,700 lines of Go
and 620 lines of C 
(for FFT-based polynomial operations, built on the FLINT library~\cite{flint}).
Unless noted otherwise, our evaluations use an FFT-friendly 87-bit field.
Our servers communicate 
with each other using Go's TLS implementation.
Clients encrypt and sign their messages to servers using NaCl's
``box'' primitive, which obviates the need for client-to-server TLS connections.
Our code is available online at \url{https://crypto.stanford.edu/prio/}.

We evaluate the SNIP-based variant of \name (Section~\ref{sec:disrupt:overview}) 
and also the variant in which the servers keep the $\Ver$ predicate private
(``\Name-MPC,'' Section~\ref{sec:disrupt:serv}).
Our implementation includes three optimizations described in Appendix~\ref{app:opts}.
The first uses a pseudo-random generator (e.g., AES in counter mode)
to reduce the client-to-server data transfer by a factor of roughly $s$
in an $s$-server deployment.
The second optimization allows the servers to verify SNIPs
without needing to perform expensive polynomial interpolations. 
The third optimization gives an efficient way for the servers to compute
the logical-\textsc{and} of multiple arithmetic circuits to check that
multiple $\Ver$ predicates hold simultaneously.

\begin{figure*}
\begin{minipage}[b]{0.4\textwidth}
\centering
\begingroup\makeatletter
\makeatother\endgroup \caption{\Name's SNIPs provide the robustness guarantees
  of zero-knowledge proofs at 20-50$\times$ less cost.}
\label{fig:micro}
\end{minipage}\hfill
\begin{minipage}[b]{0.27\textwidth}

\centering
\begingroup\makeatletter%
\makeatother\endgroup \caption{\Name is insensitive to the
number of aggregation servers.}
\label{fig:servers}
\end{minipage}\hfill
\begin{minipage}[b]{0.27\textwidth}
\centering
\begingroup\makeatletter%
\makeatother\endgroup \caption{\Name's use of SNIPs (\S\ref{sec:disrupt}) reduces bandwidth consumption.}
\label{fig:server-bw}
\end{minipage}
\end{figure*}

We compare \name against a private aggregation scheme that uses
non-interactive zero-knowledge proofs (NIZKs) to provide
robustness.
This protocol is similar to the ``cryptographically verifiable''
interactive protocol of Kursawe et al.~\cite{kursawe2011} and
has roughly the same cost, in terms of exponentiations per client request, as the 
``distributed decryption'' variant of PrivEx~\cite{elahi2014privex}.
We implement the NIZK scheme using a Go wrapper
of OpenSSL's NIST P256 code~\cite{dediscrypto}.
We do not compare \name against systems, such as
ANONIZE~\cite{hohenberger2014anonize} and PrivStats~\cite{popa2011privacy},
that rely on an external anonymizing proxy to protect against a network
adversary.
(We discuss this related work in Section~\ref{sec:rel}.)

\subsection{Microbenchmarks}
\label{sec:eval:micro}
Table~\ref{tab:bench} presents the time required for a \name client
to encode a data submission on a workstation
(2.4 GHz Intel Xeon E5620) and mobile
phone (Samsung Galaxy SIII, 1.4 GHz Cortex A9).
For a submission of 100 integers, the client time is
roughly 0.03 seconds on a workstation, and just over
0.1 seconds on a mobile phone.

To investigate the load that \name places on the servers,
we configured five Amazon EC2 servers 
(eight-core c3.2xlarge machines, Intel Xeon E5-2680 CPUs) 
in five Amazon data centers
(N.~Va., N.~Ca., Oregon, Ireland, and Frankfurt) and had them 
run the \name protocols.
An additional three c3.2xlarge machines in 
the N.\ Va.\ data center simulated a large number of \name clients.
To maximize the load on the servers, we had each client 
send a stream of pre-generated \name data packets to the servers
over a single TCP connection.
There is no need to use TLS on the client-to-server \name connection because 
\name packets are encrypted and authenticated at the application layer
and can be replay-protected at the servers.

Figure~\ref{fig:micro} gives the throughput
of this cluster in which each client submits a vector
of zero/one integers and the servers sum these vectors. 
The ``No privacy'' line on the chart gives the throughput for a dummy 
scheme in which a single server accepts encrypted client data submissions
directly from the clients with no privacy protection whatsoever.
The ``No robustness'' line on the chart gives the throughput for a 
cluster of five servers that use a secret-sharing-based private
aggregation scheme ({\em \`a la} Section~\ref{sec:simple}) with no robustness protection.
The five-server ``No robustness'' scheme is slower than the
single-server ``No privacy'' scheme because of the cost of
coordinating the processing of submissions amongst the five servers.
The throughput of \name is within a factor of 5$\times$ of the no-privacy
scheme for many submission sizes, and \name outperforms the NIZK-based
scheme by more than an order of magnitude. 

Figure~\ref{fig:servers} shows how the throughput of a \name
cluster changes as the number of servers increases, when the system
is collecting the sum of 1,024 one-bit client-submitted integers,
as in an anonymous survey application.
For this experiment, we locate all of the servers in the same data center,
so that the latency and bandwidth between each pair of servers is roughly constant.
With more servers, an adversary has to compromise a
larger number of machines to violate \name's privacy guarantees.
\begin{figure*}
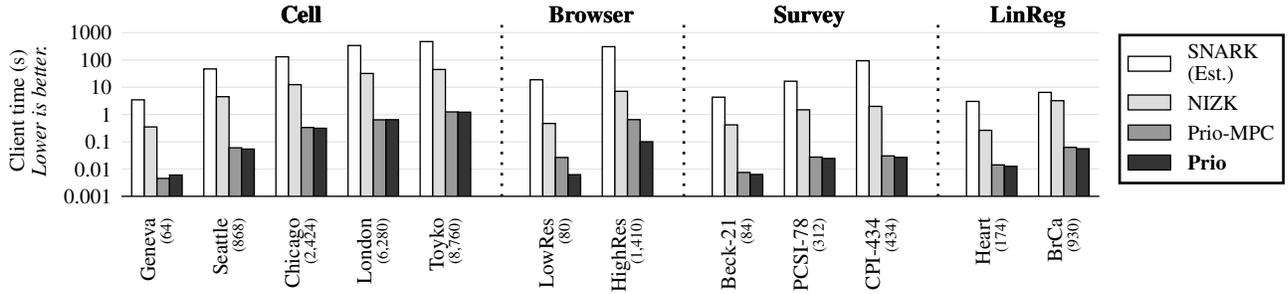

\centering
\begingroup\makeatletter
\makeatother\endgroup \caption{Client encoding time for
different application domains when using \name, a non-interactive
zero-knowledge system (NIZK), or a SNARK-like system (estimated).
Averaged over eight runs.
The number of $\times$ gates in the $\Ver$ circuit is listed
in parentheses.}
\label{fig:geo-time}
\end{figure*}

Adding more servers barely affects the system's throughput.
The reason is that we are able to load-balance the bulk of the work of checking
client submissions across all of the servers. (This optimization is 
only possible because we require robustness to hold only if all servers are honest.)
We assign a single \name server to be the ``leader'' that coordinates
the checking of each client data submission.
In processing a single submission in an $s$-server cluster, the leader 
transmits $s$ times more bits than a non-leader, but as the number of
servers increases, each server is a leader for a smaller share of incoming submissions.
The NIZK-based scheme also scales well: as the number of servers increases, the
heavy computational load of checking the NIZKs is distributed over more machines.

Figure~\ref{fig:server-bw} shows the number of bytes
each non-leader \name server needs to transmit to check the validity of a single
client submission for the two \name variants, and for the NIZK scheme.
The benefit of \name is evident: the \name servers transmit
a constant number of bits per submission---{\em independent} of the size of
the submission or complexity of the $\Ver$ routine.
As the submitted vectors grow,
\name yields a 4,000-fold bandwidth saving over NIZKs,
in terms of server data transfer.

\subsection{Application scenarios}

To demonstrate that Prio's data types are expressive enough to collect
real-world aggregates, we have configured \name 
for a few potential application domains.
\sitpara{Cell signal strength.}
A collection of \name servers can collect the average mobile signal strength in
each grid cell in a city without leaking the user's location history to the aggregator.
We divide the geographic area into a km$^2$ grid---the number of grid cells
depends on the city's size---and we encode the signal strength
at the user's present location as a four-bit integer.
(If each client only submits signal-strength data for a few
grid cells in each protocol run, extra optimizations can 
reduce the client-to-server data transfer.
See ``Share compression'' in Appendix~\ref{app:afe-extra}.)

\sitpara{Browser statistics.}
The Chromium browser uses the RAPPOR system to gather private information about
its users~\cite{erlingsson2014rappor,chromium2016}.
We implement a \name instance for gathering a subset of these statistics:
average CPU and memory usage, along with the frequency counts of 
16 URL roots.
structure, described in Appendix~\ref{app:afe-extra}.
We experiment with both low- and high-resolution
parameters ($\delta = 2^{-10}$, $\epsilon = 1/10$;
$\delta = 2^{-20}$, $\epsilon = 1/100$).

\sitpara{Health data modeling.}
We implement the \afe for training a regression
model on private client data.
We use the features from 
a preexisting heart disease data set (13 features of varying types: age, sex,
cholesterol level, etc.)~\cite{heart-data} and 
a breast cancer diagnosis data set (30 real-valued features using
14-bit fixed-point numbers)~\cite{breast-data}.

\sitpara{Anonymous surveys.}
We configure \name to compute aggregates responses to sensitive surveys:
we use the Beck Depression Inventory (21 questions on a 1-4 scale)~\cite{beck}, 
the Parent-Child Relationship Inventory (78 questions on a 1-4 scale)~\cite{pcri}, and
the California Psychological Inventory (434 boolean questions)~\cite{cpi}.

\paragraph{Comparison to alternatives.}
In Figure~\ref{fig:geo-time}, we compare the computational
cost \name places on the client to the costs of other
schemes for protecting robustness against misbehaving clients,
when we configure the system for the aforementioned applications.
The fact that a \name client need only perform a single
public-key encryption means that it dramatically outperforms
schemes based on public-key cryptography.
If the $\Ver$ circuit has $M$ multiplication gates, 
producing a discrete-log-based NIZK requires the client to perform
$2M$ exponentiations (or elliptic-curve point multiplications).
In contrast, \name requires $O(M\log M)$ multiplications in a relatively
small field, which is much cheaper for practical values of~$M$.

\begin{figure*}
  \begin{minipage}[b]{0.42\textwidth}
      \centering
\begingroup\makeatletter
\makeatother\endgroup \caption{
  Time for a client to encode a submission
  consisting of $d$-dimensional training example of
  14-bit values for computing a private least-squares regression.
}
\label{fig:linreg}
\end{minipage}~~
\begin{minipage}[b]{0.55\textwidth}
\definecolor{Gray}{gray}{0.9}
\newcolumntype{g}{>{\columncolor{Gray}}c}
\centering
{\footnotesize
  \begin{tabular}{r c c c c c g}
    &\multicolumn{1}{c}{{No privacy}} & \multicolumn{2}{c}{{No robustness}} &\multicolumn{3}{c}{\textbf{\Name}}\\ \hline \hhline{~~~~---}    $d$    & Rate  &  Rate & Priv.~cost & Rate & Robust.~cost            &Tot.~cost\\ \hline 
  $2$ & 14,688 &2,687 & $5.5\times$   &2,608  & $1.0\times$          &$\ \,\,5.6\times$\\
  $4$ & 15,426 &2,569 & $6.0\times$   &2,165 &  $1.2\times$          &$\ \,\,7.1\times$\\
  $6$ & 14,773 &2,600 & $5.7\times$   &2,048 & $1.3\times$          &$\ \,\,7.2\times$\\
  $8$ & 15,975 &2,564 & $6.2\times$   &1,606 & $1.6\times$          &$\ \,\,9.5\times$\\
  $10$ & 15,589& 2,639 & $5.9\times$  &1,430 & $1.8\times$          &$10.9\times$\\
  $12$ & 15,189& 2,547 & $6.0\times$  &1,312 & $1.9\times$          &$11.6\times$\\\hline
\end{tabular}
}
\captionof{table}{The throughput, in client requests per second, of a global five-server cluster running a private $d$-dim.~regression.
  We compare a scheme with no privacy,
  with privacy but no robustness,
  and \name (with both).}
\label{tab:linreg}
\end{minipage}
\end{figure*}

In Figure~\ref{fig:geo-time}, we give conservative estimates of 
the time required to generate a zkSNARK proof, based on timings of
libsnark's~\cite{ben2014succinct} 
implementation of the Pinocchio system~\cite{parno2013pinocchio} 
at the 128-bit security level.
These proofs have the benefit of being very short: 288 bytes, 
irrespective of the complexity of the circuit.
To realize the benefit of these succinct proofs, the statement
being proved must also be concise since the verifier's running time 
grows with the statement size.
To achieve this conciseness in the \name setting would require
computing $sL$ hashes ``inside the SNARK,'' with 
$s$ servers and submissions of length~$L$.

We optimistically estimate that each hash computation requires only 300
multiplication gates, using a subset-sum hash function~\cite{ajtai1996generating,ben2014scalable,goldreich1996collision,impagliazzo1996efficient}, 
and we ignore the cost of computing the $\Ver$ circuit in the SNARK.
We then use the timings from the libsnark paper to arrive at the cost estimates.
Each SNARK multiplication gate requires the client to compute a number of
exponentiations, so the cost to the client is large, though the 
proof is admirably short.

\subsection{Machine learning}
Finally, we perform an end-to-end evaluation of \name
when the system is configured to train a $d$-dimensional
least-squares regression model on private client-submitted data,
in which each training example consists of a vector of 14-bit integers.
These integers are large enough to represent vital health information, for example.

In Figure~\ref{fig:linreg}, we show the client encoding cost 
for \name, along with the no-privacy and no-robustness schemes described
in Section~\ref{sec:eval:micro}.
The cost of \name's privacy and robustness guarantees amounts to roughly
a 50$\times$ slowdown at the client over the 
no-privacy scheme due to the overhead of the
SNIP proof generation.
Even so, the absolute cost of \name to the client is small---on the order of one
tenth of a second.

Table~\ref{tab:linreg} gives the rate at which the globally 
distributed five-server cluster described in Section~\ref{sec:eval:micro} 
can process client submissions with and without privacy and robustness.
The server-side cost of \name is modest:
only a 1-2$\times$ slowdown over the no-robustness scheme,
and only a 5-15$\times$ slowdown over a scheme with
{\em no privacy at all}. 
In contrast, the cost of robustness for the state-of-the-art NIZK schemes, 
per Figure~\ref{fig:micro}, is closer to 100-200$\times$.

 \section{Discussion}
\label{sec:disc}

\paragraph{Deployment scenarios.}
\Name ensures client privacy as long as at least 
one server behaves honestly. 
We now discuss a number of deployment scenarios in which this assumption
aligns with real-world incentives. 

\sitpara{Tolerance to compromise.}
\Name lets an organization compute aggregate data about its clients
without ever storing client data in a single vulnerable location.   
The organization could run all $s$ \name servers itself, which would ensures
data privacy against an attacker who compromises
up to $s-1$ servers.

\newcommand{\imgwidthrel}{0.31\textwidth}
\begin{figure*}
\centering
\subfloat[RAPPOR~\cite{erlingsson2014rappor} provides differential privacy~\cite{DP} (not information-theoretic privacy) by adding random noise to client submissions.]{\includegraphics[width=\imgwidthrel]{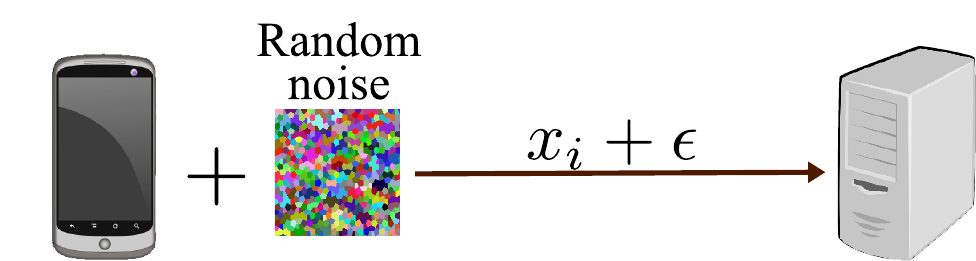}}~~~~~\subfloat[ANONIZE~\cite{hohenberger2014anonize} and PrivStats~\cite{popa2011privacy} rely on an anonymizing proxy, such as Tor~\cite{dingledine2004tor}, to protect privacy against network eavesdroppers.]{\includegraphics[width=\imgwidthrel]{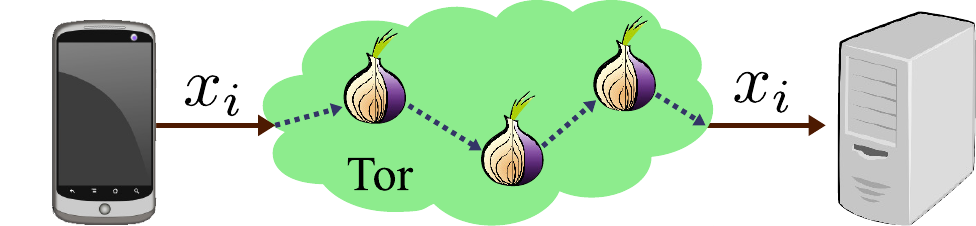}}~~~~~\subfloat[\Name and other schemes using secret sharing~\cite{melis2016,elahi2014privex,danezis2013smart,kursawe2011,jawurek2012fault,castelluccia2005efficient} offer ideal anonymity provided that the servers do not collude.]{\includegraphics[width=\imgwidthrel]{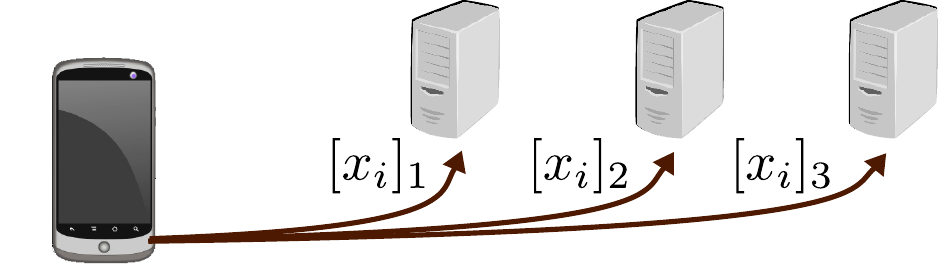}}
\caption{Comparison of techniques for anonymizing client data in private aggregation systems.}
\label{fig:network}
\end{figure*}

\sitpara{App store.}
A mobile application platform (e.g., Apple's App Store or Google's Play) can
run one \name server, and the developer of a mobile app can run the second
\name server.  This allows the app developer to collect aggregate
user data without having to bear the risks
of holding these data in the clear.

\sitpara{Shared data.}
A group of $s$ organizations could use \name to compute 
an aggregate over the union of their customers' datasets,
without learning each other's private client data. 

\sitpara{Private compute services.}
A large enterprise can contract with an external auditor
or a non-profit (e.g., the Electronic Frontier Foundation)
to jointly compute aggregate statistics over sensitive customer data using \name.

\sitpara{Jurisdictional diversity.}
A multinational organization can spread its \name servers across different
countries. If law enforcement agents seize the \name servers in one
country, they cannot deanonymize the organization's \name users.

\paragraph{Common attacks.}
\label{sec:disc:attacks}
Two general attacks apply to {\em all systems},
like \name, that produce exact (un-noised) outputs
while protecting privacy against a network adversary.
The first attack is a {\em selective denial-of-service attack}.
In this attack, the network adversary prevents all honest clients except one
from being able to contact the \name servers~\cite{serjantov2002trickle}.
In this case, the protocol output is
$f(x_\text{honest}, x_{\text{evil}_1}, \dots, x_{\text{evil}_n})$.
Since the adversary knows the $x_\text{evil}$ values, the adversary 
could infer part or all of the one honest client's private value $x_\text{honest}$.

In \name, we deploy the standard defense against this attack, which is
to have the servers wait to publish the aggregate statistic 
$f(x_1, \dots, x_n)$ until they are confident that the aggregate includes
values from many honest clients.  
The best means to accomplish this 
will depend on the deployment setting. 

One way is to have the servers keep a list of public keys of registered clients
(e.g., the students enrolled at a university).
\name clients sign their submissions with the signing key corresponding to
their registered public key and the servers wait to publish their accumulator
values until a threshold number of registered clients have submitted valid
messages. 
Standard defenses~\cite{yu2006sybilguard,yu2008sybillimit,viswanath2010analysis,alvisi2013sok}
against Sybil attacks~\cite{douceur2002sybil} would
apply here.

The second attack is an {\em intersection attack}~\cite{kedogan2002limits,berthold2002dummy,wolinsky2013hang,danezis2004statistical}.
In this attack, the adversary observes the output $f(x_1, \dots, x_n)$ of a
run of the \name protocol with $n$ honest clients.
The adversary then forces the $n$th honest client offline
and observes a subsequent protocol run, in which the servers compute $f(x'_1, \dots, x'_{n-1})$.
If the clients' values are constant over time ($x_i = x'_i$), then the adversary learns the difference
$f(x_1, \dots, x_n) - f(x_1, \dots, x_{n-1})$, which could reveal 
client $n$'s private value $x_n$ (e.g., if $f$ computes \textsc{sum}).

One way for the servers to defend against the attack is to
add differential privacy noise to the results before publishing them~\cite{DP}. 
Using existing techniques, 
the servers can add this noise in a distributed fashion to ensure that
as long as at least one server is honest,
no server sees the un-noised aggregate~\cite{dwork2006our}.
The definition of differential privacy ensures that computed
statistics are distributed approximately the same whether or not 
the aggregate includes a particular client's data. 
This same approach is also used in a system by 
Melis, Danezis, and De Cristofaro~\cite{melis2016}, which we
discuss in Section~\ref{sec:rel}. 

\paragraph{Robustness against malicious servers.}
\Name only provides robustness when all servers are honest.
Providing robustness in the face of faulty servers is obviously desirable, but we
are not convinced that it is worth the security and performance costs.
Briefly, providing robustness necessarily weakens the privacy guarantees
that the system provides: if the system protects {\em robustness} 
in the presence of $k$ faulty servers, then the system can 
protect {\em privacy} only against a coalition of at most
$s-k-1$ malicious servers.
We discuss this issue further in Appendix~\ref{app:robfault}.

\section{Related Work}
\label{sec:rel}

Private data-collection systems~\cite{elahi2014privex,castelluccia2005efficient,
danezis2013smart,jawurek2012fault,kursawe2011,duan2010,melis2016}
that use secret-sharing based methods to compute sums over private user data
typically (a) provide no robustness guarantees in the face of malicious
clients, (b) use expensive NIZKs to prevent client misbehavior, or
(c) fail to defend privacy against actively malicious servers~\cite{chen2013splitx}.

Other data-collection systems have clients send their private data to an
aggregator through a general-purpose anonymizing network, such as a
mix-net~\cite{chaum1981untraceable,danezis2004better,brickell2006efficient,kwon2015riffle}
or a DC-net~\cite{chaum1988dining,sirer2004eluding,corrigan2015riposte,corrigan2010dissent,corrigan2013proactively}.
These anonymity systems provide strong privacy properties,
but require expensive ``verifiable mixing'' techniques~\cite{bayer2012efficient,neff2001verifiable}, or
require work at the servers that is {\em quadratic} in the number of client messages
sent through the system~\cite{corrigan2015riposte,wolinsky2012dissent}.

PrivStats~\cite{popa2011privacy} and
ANONIZE~\cite{hohenberger2014anonize}
outsource to Tor~\cite{dingledine2004tor}
(or another low-latency anonymity system~\cite{le2013towards,freedman2002tarzan,reiter1998crowds})
the work of protecting privacy against a network adversary (Figure~\ref{fig:network}).
\Name protects against an adversary that can see and control the entire
network, while Tor-based schemes succumb to traffic-analysis attacks~\cite{murdoch2005low}.

In data-collection systems based on differential privacy~\cite{DP}, the client adds structured
noise to its private value before sending it to an aggregating server. 
The added noise gives the client ``plausible deniability:'' if the client sends
a value $x$ to the servers, $x$ could be the client's true private value, or it could
be an unrelated value generated from the noise.
Dwork et al.~\cite{dwork2006our}, Shi et al.~\cite{shi2011privacy},
and Bassily and Smith~\cite{bassily2015local}
study this technique in a distributed setting, and
the RAPPOR system~\cite{erlingsson2014rappor,fanti2016building},
deployed in Chromium, has put this idea into practice.
A variant of the same principle is to have a trusted proxy
(as in SuLQ~\cite{blum2005practical} and PDDP~\cite{chen2012towards}) 
or a set of minimally trusted servers~\cite{melis2016} add noise
to already-collected data.

The downside of these systems is that 
(a) if the client adds little noise, then the system does not provide much privacy, or
(b) if the client adds a lot of noise, then low-frequency events
may be lost in the noise~\cite{erlingsson2014rappor}.
Using server-added noise~\cite{melis2016} ameliorates these problems.

In theory, secure multi-party computation (MPC)
protocols~\cite{lindell2009proof,yao1986generate,goldreich1987play,ben1988completeness,beaver1990round}
allow a set of servers, with some non-collusion assumptions,
to privately compute {\em any} function over client-provided values.
The generality of MPC comes with serious bandwidth and computational costs:
evaluating the relatively simple AES circuit in an MPC
requires the parties to perform many minutes or even hours of
precomputation~\cite{damgaard2012implementing}.
Computing a function $f$ on millions of client inputs, as 
our five-server \name deployment can do in tens of
minutes, could potentially take an astronomical amount of time in a full MPC.
That said, there have been great advances in practical general-purpose MPC protocols of 
late~\cite{pinkas2009secure,bellare2013efficient,gueron2015fast,damgaard2013practical,
malkhi2004fairplay,lindell2016fast,SPDZ,mohassel2015fast,ben2008fairplaymp,bogetoft2008multiparty}.
General-purpose MPC may yet become practical
for computing certain aggregation functions that \name cannot (e.g., exact \textsc{max}),
and some special-case MPC
protocols~\cite{nikolaenko2013privacy,burkhart2010sepia,applebaum2010collaborative} 
are practical today for certain applications.

 \section{Conclusion and future work}
\Name allows a set of servers to compute aggregate statistics
over client-provided data while 
maintaining client privacy, 
defending against client misbehavior,
and performing nearly as well as
data-collection platforms that exhibit neither of these security properties.
The core idea behind \name is reminiscent of techniques used in verifiable
computation~\cite{wahby2014efficient,walfish2015verifying,williams2016strong,goldwasser2015delegating,parno2013pinocchio,gennaro2013quadratic,cormode2011verifying,ben2013snarks}, but in reverse---the client proves to a set of servers
that it computed a function correctly.
One question for future work is whether it is possible to efficiently extend \name to
support combining client encodings using a more general function than
summation, and what more powerful aggregation functions this would enable.
Another task is to investigate the possiblity of shorter SNIP
proofs: ours grow linearly in the size of the $\Ver$ circuit, 
but sub-linear-size information-theoretic SNIPs may be feasible.

\medskip

\noindent
{\small
\textbf{Acknowledgements.}
We thank the anonymous NSDI reviewers for an
extraordinarily constructive set of reviews.
Jay Lorch, our shepherd, read two drafts of this paper and
gave us pages and pages of insightful recommendations and thorough comments.
It was Jay who suggested using \name to 
privately train machine learning models, 
which became the topic of Section~\ref{sec:struct:ml}.
Our colleagues, including
David Mazi\`eres,
David J.~Wu,
Dima Kogan,
Elette Boyle,
George Danezis,
Phil Levis, 
Matei Zaharia, 
Saba Eskandarian, 
Sebastian Angel, and
Todd Warszawski
gave critical feedback that improved the content and presentation of the work.
Any remaining errors in the paper are, of course, ours alone.

This work received support from NSF, DARPA, the Simons Foundation, an
NDSEG Fellowship, and ONR. 
Opinions, findings and conclusions or recommendations expressed in this
material are those of the authors and do not necessarily reflect the views of DARPA.
\par
}

\frenchspacing
{\bibliographystyle{acm}
\bibliography{refs}
}
\nonfrenchspacing

\appendix

\section{Security definitions}
\label{app:secdefs}
\label{app:secdefs:priv}

In this section we define $f$-privacy.  
We give an informal definition, for the sake of readability,
that captures the security properties we need.
In what follows we use the standard notions of {\em negligible
  functions} and {\em computational indistinguishability} (see, e.g.,
\cite{goldreich}).  We often leave the security parameter implicit.

\begin{defn}[$f$-Privacy]
\label{defn:fpriv}
Say that there are $s$ servers in a \name deployment, 
and say that $n$ clients' values are included
in the final aggregate that the system outputs.
We say that the scheme provides $f$-{\em privacy} for a function $f$, if
for:
\begin{itemize}
\item every subset of at most $s-1$ malicious servers, and
\item every number of malicious clients $m \leq n$, 
\end{itemize} 
there exists an efficient simulator that,
for every choice of the honest clients' inputs $(x_1, \dots, x_{n-m})$,
takes as input:
\begin{itemize} 
\item the public parameters to the protocol run 
      (all participants' public keys, the description of the aggregation 
      function $f$, the cryptographic parameters, etc.),
\item the indices of the adversarial clients and servers, 
\item oracle access to the adversarial participants, and
\item the value $f(x_1, \dots, x_{n-m})$, 
\end{itemize} 
and outputs a simulation of the adversarial participants' view of the protocol
run whose distribution is computationally indistinguishable 
from the distribution of the adversary's view of the real protocol run.
\end{defn}

By engaging in the protocol, the adversary learns the 
value $f(x_1, \dots, x_{n-m})$ exactly.
It is therefore critical that the honest servers ensure that the number $n-m$ of
honest clients' values included in the aggregate is ``large enough.'' 
The honest servers must use out-of-band means to ensure that many honest clients'
values are included in the final aggregate.
See discussion of this and related issues in Section~\ref{sec:disc}.

Let $\textsf{SORT}$ be the
function that takes $n$ inputs and outputs them in 
lexicographically increasing order.

\begin{defn}[Anonymity]  \label{def:anon}
We say that a data-collection scheme provides {\em anonymity} if 
it provides $f$-privacy, 
in the sense of 
Definition~\ref{defn:fpriv},
for $f = \textsf{SORT}$.
\end{defn}

A scheme that provides this form of anonymity leaks to the adversary
the entire list of honest clients' inputs $(x_1, \dots, x_{n-m})$, but the adversary
learns nothing about which client submitted which value $x_i$.
For example, if each client submits their location via a data-collection
scheme that provides anonymity, the servers learn the list of submitted locations
$\{\ell_1, \dots, \ell_{n-m}\}$, 
but the servers learn nothing about whether honest client $x$ or $y$ is in a
particular location $\ell^*$.

\begin{defn}
\label{defn:symm}
A function $f(x_1, \dots, x_n)$ is {\em symmetric} if, for all permutations
$\pi$ on $n$ elements, the equality $f(x_1, \dots, x_n) = f(x_{\pi(1)}, \dots, x_{\pi(n)})$ holds.
\end{defn}

\begin{claim}
\label{claim:anonpriv}
Let $\D$ be a data-collection scheme that provides $f$-privacy, in the sense of
Definition~\ref{defn:fpriv}, for a symmetric function $f$.
Then $\D$ provides anonymity.
\end{claim}

\begin{proof}[Proof sketch]
The fact that $\D$ provides $f$-privacy implies the existence of a simulator
$S_\D$ that takes as input $f(x_1, \dots, x_{n-m})$, along with other public values,
and induces a distribution of protocol transcripts indistinguishable from the real one.
If $f$ is symmetric, $f(x_1, \dots, x_{n-m}) = f(x'_1, \dots, x'_{n-m})$, where
\[ (x'_1, \dots, x'_{n-m}) = \textsf{SORT}(x_1, \dots, x_{n-m}).\] 
Using this fact, we construct the simulator required for the anonymity definition:
on input $(x'_1, \dots, x'_{n-m}) = \textsf{SORT}(x_1, \dots, x_{n-m})$, compute $f(x'_1, \dots, x'_{n-m})$,
and feed the output of $f$ to the simulator $S_\D$.
The validity of the simulation is immediate.
\end{proof}

The following claim demonstrates that it really only makes sense to use
an $f$-private data collection scheme when the function $f$ is symmetric,
as all of the functions we consider in \name are.

\begin{claim}
\label{claim:nonsym}
Let $f$ be a non-symmetric function.  Then there is no anonymous data
collection scheme that correctly computes~$f$.
\end{claim}

\begin{proof}[Proof sketch]
Let there be no malicious clients ($m = 0$).
Because $f$ is not symmetric, there exists an input
$(x_1,\ldots,x_n)$ in the domain of $f$, and a permutation $\pi$ on~$n$ 
elements, such that $f(x_1,\ldots,x_n) \neq f(x_{\pi(1)},\ldots,x_{\pi(n)})$.

Let $\D$ be a data-collection scheme that implements the 
aggregation function $f(x_1, \dots, x_n)$.  
This $\D$ outputs $f(x_1, \ldots, x_n)$ for all $x_1,\ldots,x_n$ in the domain,
and hence $f(x_1,\ldots,x_n)$ is part of the protocol transcript.

For $\D$ to be anonymous, there must be a simulator 
that takes $\textsf{SORT}(x_1,\ldots,x_n)$ as input, and simulates
the protocol transcript. In particular, it must output $f(x_1, \ldots, x_n)$.
But given $\textsf{SORT}(x_1,\ldots,x_n)$, it will necessarily fail to
output the correct protocol transcript on either
$(x_1,\ldots,x_n)$ or $(x_{\pi(1)},\ldots,x_{\pi(n)})$.
\end{proof}

\paragraph{Robustness.}
Recall that each \name client holds a value $x_i$, where the value $x_i$
is an element of some set of data items~$\D$.
For example, $\D$ might be the set of $4$-bit integers.
The definition of robustness states that when all servers are honest,
a set of malicious clients cannot influence the final aggregate, beyond
their ability to choose arbitrary {\em valid} inputs.
For example, malicious clients can choose arbitrary $4$-bit integers
as their input values, but cannot influence the output in any other way. 

\begin{defn}[Robustness]  \label{def:robust}
Fix a security parameter $\lambda > 0$.
We say that an $n$-client \name deployment provides \textit{robustness} if, 
when all \name servers execute the protocol faithfully, 
for every number $m$ of malicious clients (with $0 \leq m \leq n$), and
for every choice of honest client's inputs $(x_1, \dots, x_{n-m}) \in \D^{n-m}$,
the servers, with all but negligible probability in $\lambda$,
output a value in the set:
\[ \Big\{ f(x_1, \dots, x_n) \mid (x_{n-m+1}, \dots, x_n) \in \D^m \Big\}. \]
\end{defn}

 \section{Robustness against faulty servers}
\label{app:robfault}
If at least one of the servers is honest, \name ensures
that the adversary learns nothing about clients' data, except
the aggregate statistic.
In other words, \name ensures {\em client privacy} if at least one of servers
is honest.

\name provides {\em robustness} only if all servers are honest.
Providing robustness in the face of faulty servers is obviously desirable, 
but we
are not convinced that it is worth the security and performance costs.
First, providing robustness necessarily weakens the privacy guarantees
that the system provides: if the system protects {\em robustness} 
in the presence of $k$ faulty servers, then the system can 
protect {\em privacy} only against a coalition of at most
$s-k-1$ malicious servers.
The reason is that, if robustness 
holds against $k$ faulty servers, then $s-k$ honest servers must 
be able to produce a correct output even if these $k$ faulty servers are offline.
Put another way: $s-k$ {\em dishonest} servers can recover the output of the system
even without the participation of the $k$ honest servers.
Instead of computing an aggregate over many clients ($f(x_1, \dots, x_n)$), 
the dishonest servers can compute the ``aggregate'' over a single client's
submission ($f(x_1)$) and essentially learn that client's private data value.

So strengthening robustness in this setting weakens privacy.
Second, protecting robustness comes at a performance cost:
some our optimizations use a ``leader'' server
to coordinate the processing of each client submission (see Appendix~\ref{app:opts}).
A faulty leader cannot compromise privacy, but \textit{can} compromise
robustness.
Strengthening the robustness property would force us to abandon these optimizations.

That said, it would be possible to extend \name to provide
robustness in the presence of corrupt servers using standard techniques~\cite{beaver1991secure}
(replace $s$-out-of-$s$ secret sharing with Shamir's
threshold secret-sharing scheme~\cite{shamir1979share}, etc.).

 \section{MPC background}
\label{app:beaver}

This appendix reviews the definition of arithmetic circuits
and Donald Beaver's multi-party 
computation protocol~\cite{beaver1991efficient}.

\subsection{Definition: Arithmetic circuits}
\label{app:beaver:circuit}
An \textit{arithmetic circuit} $\C$ over a finite field
$\F$ takes as input a vector 
$x = \langle x^{(1)}, \dots, x^{(L)} \rangle \in \F^L$
and produces a single field element as output.
We represent the circuit as a directed acyclic graph, in which
each vertex in the graph is either an {\em input}, a {\em gate}, or an {\em
output} vertex.

Input vertices have in-degree zero and 
are labeled with a variable in $\{x^{(1)}, \dots, x^{(L)} \}$ or a constant in $\F$.
Gate vertices have in-degree two and 
are labeled with the operation $+$ or $\times$.
The circuit has a single output vertex, which has out-degree zero. 

To compute the circuit $\C(x) = \C(x^{(1)}, \dots, x^{(L)})$, 
we walk through the circuit from
inputs to outputs, assigning a value in $\F$ to each wire 
until we have a value on the output wire, which is the 
value of $\C(x)$.
In this way, the circuit implements a mapping $\C : \F^L \to \F$.

\subsection{Beaver's MPC protocol}
\label{app:beaver:proto}

This discussion draws on the clear exposition by Smart~\cite{smart2011fhe}.

Each server starts the protocol holding a share $[x]_i$
of an input vector $x$.
The servers want to compute $\C(x)$, for some arithmetic circuit $\C$.

The multi-party computation protocol walks through the 
circuit $\C$ wire by wire, from inputs to outputs. 
The protocol maintains the invariant that, at the $t$-th time step, 
each server holds a share of the value on the $t$-th wire in the circuit.
At the first step, the servers hold shares of the input wires (by construction)
and in the last step of the protocol, the servers hold shares of the output wire.
The servers can then publish their shares of the output wires, which 
allows them all to reconstruct the value of $\C(x)$.
To preserve privacy, no subset of the servers must ever have enough information
to recover the value on any internal wire in the circuit.

There are only two types of gates in an arithmetic circuit (addition gates and
multiplication gates), so we just have to show how the servers can
compute the shares of the outputs of these gates from shares of the inputs.
All arithmetic in this section is in a finite field $\F$.

\smallskip
\noindent
\textit{Addition gates.} 
In the computation of an addition gate ``$y + z$'', the $i$th server holds shares 
$[y]_i$ and $[z]_i$ of the input wires and the server needs to compute a share 
of $y+z$.
To do so, the server can just add its shares locally
\[ [y+z]_i = [y]_i + [z]_i. \]

\smallskip
\noindent
\textit{Multiplication gates.} 
In the computation of a multiplication gate, the $i$th server 
holds shares $[y]_i$ and $[z]_i$ and wants to compute
a share of $y z$.

When one of the inputs to a multiplication gate is a constant, 
each server can locally compute a share of the output of the gate.
For example, to multiply a share $[y]_i$ by a constant $A \in \F$, each
server $i$ computes their share of the product as $[Ay]_i = A[y_i]$.

Beaver showed that the servers can use pre-computed {\em multiplication triples} 
to evaluate multiplication gates~\cite{beaver1991efficient}. 
A multiplication triple is a one-time-use triple of values $(a, b, c) \in \F^3$, 
chosen at random subject to the constraint that $a \cdot b = c \in \F$.
When used in the context of multi-party computation, each server~$i$ 
holds a share $([a]_i, [b]_i, [c]_i) \in \F^3$ of the triple.

Using their shares of one such triple $(a, b, c)$, the servers can jointly
evaluate shares of the output of a multiplication gate $yz$.
To do so, each server $i$ uses her shares $[y]_i$ and $[z]_i$ of the input wires,
along with the first two components of its multiplication triple to compute 
the following values:
\begin{align*}
[d]_i = [y]_i - [a]_i \quad;\quad [e]_i = [z]_i - [b]_i.
\end{align*}
Each server $i$ then broadcasts $[d]_i$ and $[e]_i$.
Using the broadcasted shares,
every server can reconstruct $d$ and $e$ and can compute:
\begin{align*}
\sigma_i = de/s + d[b]_i + e[a]_i + [c]_i.
\end{align*}
Recall that $s$ is the number of servers---a public constant---and
the division symbol here indicates division
(i.e., inversion then multiplication) 
in the field $\F$.
A few lines of arithmetic confirm that $\sigma_i$ is
a sharing of the product $yz$.
To prove this we compute:
\begin{align*}
\sum_i \sigma_i &= \sum_i \big ( de/s + d[b]_i + e[a]_i + [c]_i \big )\\
 &= de + db + ea + c\\
 &= (y - a)(z-b) + (y-a)b + (z-b)a + c\\
 &= (y - a)z + (z-b)a + c\\
 &= yz - az + az - ab + c\\
 &= yz - ab + c\\
 &= yz. 
\end{align*}
The last step used that $c = ab$ (by construction of the multiplication triple), so:
$\sum_i \sigma_i = yz$, which implies that $\sigma_i = [yz]_i$.

Since the servers can perform addition and multiplication of shared values,
they can compute any function of the client's data value
in this way, as long as they
have a way of securely generating multiplication triples.
The expensive part of traditional MPC protocols is the process by which
mutually distrusting servers generate these triples in a distributed way.

 \section{SNIP security analysis}
\label{app:zkproof}

\subsection{Soundness}
\label{app:zkproof:sound}

We prove here that the SNIP construction of Section~\ref{sec:disrupt}
is sound.
We define soundness of the SNIP construction using the following 
experiment.
The adversary's task is to produce an input $x$ and
a SNIP proof $\pi = (f(0), g(0), h, a, b, c)$ shared amongst the
servers, such that $\Ver(x) \neq 1$ and yet the servers accept 
$x$ as valid.
Assume, without loss of generality,
that the output of the $\Ver$ circuit is
the value on the output wire of the last multiplication gate
in the circuit, in topological order from inputs to outputs.

\newcommand*{\tsum}{\textstyle{\sum}}
\begin{framed}
  \noindent
  \textsf{Soundness}.\ \ In the experiment, the adversary $\A$ 
      plays the role of a dishonest client 
      interacting with honest servers.
  \begin{enumerate}
    \item Run the adversary $\A$. For each server $i$, the
          adversary outputs a set of values:
        \begin{itemize}
          \item $[x]_i \in \F^L$,
          \item $([f(0)]_i, [g(0)]_i) \in \F^2$, 
          \item $[h]_i \in \F[x]$ of degree at most $2M$, and
          \item $([a]_i, [b]_i, [c]_i) \in \F^3$.
        \end{itemize}
      \item     Choose a random point $r \gets_R \F$.
          For each server $i$, 
          compute $[f]_i$ and $[g]_i$ as in the real protocol, and
          evaluate $[f(r)]_i$, $[r\cdot g(r)]_i$, and $[r \cdot h(r)]_i$.
        \item Define the following values, where $s$ is a constant
              representing the number of servers:
          \begin{align*}
            x &= \tsum_i [x]_i &           a &= \tsum_i[a]_i\\
            f(r) &= \tsum_i [f(r)]_i &     b &= \tsum_i [b]_i\\
            r \cdot g(r) &= \tsum_i [r\cdot g(r)]_i &     c &= \tsum_i[c]_i\\
            h(M) &= \tsum_i [h]_i(M) &     d &= f(r) - a \\
                 &&                        e &= r \cdot g(r) - b
          \end{align*}
          \begin{align*}
          \sigma  &= \tsum_i \big(de/s + d[b]_i + e[a]_i + [c]_i - [r\cdot h(r)]_i \big)\\
                  &= de + db + ea + c - h(r)\\
                  &= (f(r)-a)(r\cdot g(r) - b) + (f(r) - a)b \\
                  &\quad\qquad+ (r\cdot g(r)-b)a + c - r\cdot h(r)\\
                  &= r \cdot (f(r) g(r) - h(r)) + (c - ab).
          \end{align*}

  \item We say that the adversary wins the game if
        $$h(M) = 1,\ \  \sigma = 0,\ \  \text{and}\ \ \Ver(x) \neq 1.$$ 

\end{enumerate}
\end{framed}
Let $W$ be the event that the adversary wins the 
soundness game. 
We say that the SNIP protocol is sound if, for
all adversaries $\A$, even computationally unbounded ones,
$$\Pr[W] \leq (2M+1)/|\F|,$$
where the probability is over the uniform choice of $r$ in $\F$
in Step~2.
Here $M$ is the number of multiplication gates in the
$\Ver$ circuit.  

\begin{thm}
  The SNIP protocol of Section~\ref{sec:disrupt} is sound.
\end{thm}

\begin{proof}
First, consider the probability that $\sigma = 0$,
over the choice of $r$ in $\F$.
We can express the value of $\sigma$ as the evaluation of
the following polynomial $P(t)$ at the point $r$:
\begin{align*}
P(t) &= t \cdot (f(t) \cdot g(t) - h(t)) + (c - ab)\\
    &= t \cdot Q(t) + (c - ab),
\end{align*}
for some polynomial $Q(t)$.
If $f \cdot g \neq h$, then $Q$ is non-zero and the polynomial $P$ is a
non-zero polynomial of degree at most $2M + 1$, no matter 
how $Q$ and $(c - ab)$ are related.
Furthermore, we know that the
choice of $r$ is independent of
$(a, b, c)$ and $Q = f \cdot g - h$, since the adversary must produce
these values in the experiment before $r$ is chosen.
The choice of $r$ is then independent of the polynomial $P$.
Since $P$ is a non-zero polynomial of degree at most $(2M+1)$, $P$ can have at most 
$2M+1$ zeros in $\F$.
This implies that the probability that $P(r) = \sigma = 0$, for
a random $r$ independent of $P$, is at most $(2M+1)/|\F|$.
So, whenever $f \cdot g \neq h$,
$\Pr[W] \leq (2M+1)/|\F|$.

Let us now consider the case when $f \cdot g = h$.
Label the gates in the $\Ver$ circuit in topological 
order from inputs to outputs.
\begin{claim}\label{claim:mulgate}
  When $f \cdot g = h$, the value $h(t)$ is equal
  to the output of the value on the output wire
  of the $t$-th multiplication gate in the $\Ver(x)$
  circuit.
\end{claim}
\begin{proof}
  By induction on $t$.

  \paragraph{Base case $(t=1)$.}
  The left and right inputs to the first multiplication gate in the
  circuit are affine functions $A_{L,1}$ and $A_{R,1}$ of the input $x$.
  (These functions could also be constant functions.)
  In constructing the polynomials $f$ and $g$ in the experiment,
  we set:
  \[ f(1) = A_{L,1}(x) \qquad \text{and} \qquad g(1) = A_{R,1}(x).\]
  Since $f\cdot g = h$, we have that
  \[h(1) = f(1) \cdot g(1) = A_{L,1}(x) \cdot A_{R,1}(x),\] and 
  $h(1)$ is truly the value on the output wire 
  of the first multiplication gate in the
  $\Ver(x)$ circuit.

  \paragraph{Induction step.}
  The left and right inputs to the $t$th 
  multiplication gate are affine functions
  $A_{L,t}$ and $A_{R,t}$ of $x$ and the outputs of 
  all prior multiplication gates in the circuit:
  $(h(1), \dots, h(t-1))$.
  By the induction hypothesis, the values $(h(1), \dots, h(t-1))$
  are the values on the output wires of the first $t-1$ multiplication
  gates in the $\Ver$ circuit.
  In constructing $f$ and $g$, we set:
  \begin{align*}
    f(t) &= A_{L,t}(x, (h(1), \dots, h(t-1))\\
    g(t) &= A_{R,t}(x, (h(1), \dots, h(t-1)).
  \end{align*}
  These values are exactly the values on 
  the left and right input wires to the
  $t$th multiplication gate in the $\Ver$ circuit.
  Finally, since $f\cdot g = h$, we have that
  ${h(t) = f(t) \cdot g(t)}$, so the value $h(t)$ is equal to the value
  on the output wire of the $t$th multiplication gate in the
  $\Ver$ circuit.
\end{proof}
Now, by Claim~\ref{claim:mulgate}, when $f\cdot g = h$,
the value of $h(M)$ is the output of the $M$th multiplication
gate in the $\Ver(x)$ circuit, which is actually the value
on the output wire of the circuit.
(Here we assume without loss of generality 
that the output value of the last multiplication
gate is value that the circuit outputs.)
The probability that $h(M) = 1$ when $\Ver(x) \neq 1$ is
then zero.
So, whenever $f \cdot g = h$, 
$\Pr[W] = 0$. 

We have shown that $\Pr[W] \leq (2M + 1)/|\F|$ no matter whether 
$f \cdot g = h$.
This concludes the proof of the theorem.
\end{proof}

\subsection{Zero knowledge property}
\label{app:zkproof:zero}

In this section, we show that as long as one server
is honest, an adversary who controls the remaining 
servers learns nothing about the client's data
by participating in the SNIP verification protocol.
Informally, we show that for all 
$x \in \F^L$ such that $\Ver(x) = 1$, 
an adversary who controls
a proper subset of the servers 
can perfectly simulate its interaction with the honest parties
(the client and the honest servers). 
The simulation is generated without knowledge of the client's data~$x$,
proving that the adversary learns nothing about the client's data
from the protocol. 

We first precisely define the notion of security that
we require from our SNIP construction.
We define a ``real'' world experiment, in which the adversary is given
data as in the real protocol,
and an ``ideal'' world experiment, in which the
adversary is given data generated by a simulator. 
We say that the SNIP protocol maintains the zero-knowledge
property if the adversary---even if the adversary receives $x$ as input---has no advantage at distinguishing these two experiments.

Since the role of all servers is symmetric, and since all values exchanged in
the protocol are shared using additive secret sharing, it is without loss of
generality to prove security for the case of a single adversarial server
interacting with a single honest server.
For a value $v \in \F$ we use $\Ash{v}$ and $\Hsh{v}$ to denote the 
adversarial and honest parties'
shares of $v$.  Then $\Ash{v} + \Hsh{v} = v \in \F$.

In what follows, 
we denote
the single Beaver triple that the servers use to execute the protocol
by $(a, b, c) \in \F^3$.
We let $r \in \F$ be the point 
that the servers use to execute the polynomial identity test.
We assume that the $\Ver$ circuit has $M$ multiplication gates,
and that the output of the $M$th multiplication gate is the 
value that the circuit outputs.

\begin{framed}
\noindent
\textsf{Real-world experiment.}\ \ 
The experiment is parameterized by an
input $x \in \F^L$ and a point $r \in \F$.
The experiment is with 
an adversary $(\A_1, \A_2)$, 
representing the malicious server.

\begin{enumerate}[noitemsep,nolistsep]
  
    \item Choose random $(a, b, c) \in \F^3$, subject
        to the constraint that $a \cdot b = c$.
      Construct polynomials $f$, $g$, and $h$ over $\F$
            as in the real protocol.
      Compute $f(r)$, $g(r)$, and $h(r)$.
      Split the values $a$, $b$, $c$, $x$, $f(0)$, $g(0)$, and $h$ into shares and
      compute
      \begin{align*}
      \Hsh{d} &= \Hsh{f(r)} - \Hsh{a} \quad\in \F \\
      \Hsh{e} &= \Hsh{r \cdot g(r)} - \Hsh{b} \quad\in \F \\
      \textsf{msg}_1 &= \big( \Ash{a}, \Ash{b}, \Ash{c}, \Ash{x}, \Ash{f(0)}, \Ash{g(0)}, \\
                  &\qquad\qquad\Ash{h}, \Hsh{h(M)}, \Hsh{d}, \Hsh{e}, \big).
    \end{align*}
    Note that $\Hsh{h(M)}$ is the honest server's share of the output wire
    of the $\Ver(x)$ circuit.

  \item 
      Run the first-stage adversary $\A_1$:
\[
  (\state_\A, \Ash{\hat{d}}, \Ash{\hat{e}}) \gets \A_1(x, r, \textsf{msg}_1).
\]

\item Compute the honest parties' response by setting
  $\hat{d} = \Ash{\hat{d}} + \Hsh{d}$ and
  $\hat{e} = \Ash{\hat{e}} + \Hsh{e}$, and computing
  \[\Hsh{\hat{\sigma}} = \hat{d}\hat{e}/2 + \hat{d}\Hsh{b} + \hat{e}\Hsh{a} + \Hsh{c} - \Hsh{r \cdot h(r)} \in \F.\]
\item Output \quad
      $\hat{\beta} \gets \A_2(\state_\A, \Hsh{\hat{\sigma}}) \ \ \in \{0,1\}$.
\end{enumerate}
\end{framed}

\begin{framed}
\noindent
\textsf{Ideal-world experiment.}\ \ 
This experiment is
parameterized by an
input $x \in \F^L$ and a point $r \in \F$.
The experiment is with 
an adversary $(\A_1, \A_2)$ and a simulator $S = (S_1, S_2)$.
The simulator is only given $r \in \F$ as input.

\begin{enumerate}\item Run the simulator:\ \  $(\state_S, \textsf{msg}_1) \gets S_1(r)$.
\item 
  $(\state_\A, \Ash{\hat{d}}, \Ash{\hat{e}}) \gets \A_1\big(x, r, \textsf{msg}_1 \big)$. 
\item Run the simulator:\\
  $\Hsh{\hat{\sigma}} \gets S_2(\state_S, \Ash{\hat{d}}, \Ash{\hat{e}})$.
\item Output $\hat{\beta} \gets \A_2(\state_\A, \Hsh{\hat{\sigma}}) \in \{0,1\}$.
\end{enumerate}
\end{framed}

We let $\textrm{Exp}_\textrm{real}(x,r)$
and $\textrm{Exp}_\textrm{ideal}(x,r)$
be a random variable indicating the adversary's output in
the real-world and ideal-world experiments, respectively.

We say that the SNIP protocol is zero knowledge if no adversary
$\A$, even a computationally unbounded one, 
can distinguish the real-world experiment from the
ideal-world experiment.
That is, for all $x \in \F^L$ such that
$\Ver(x) = 1$, and for all $r \not \in \{1, \dots, M\} \in \F$, 
we require that:
\begin{equation} \label{eq:snipzk}
\big|\Pr[ \textrm{Exp}_\textrm{real}(x,r)  = 1]  - 
       \Pr[ \textrm{Exp}_\textrm{ideal}(x,r) = 1] \big| = 0. 
\end{equation}

If the SNIP protocol satisfies this notion of security, even if the
adversary has a candidate guess $x^* \in \F$ for
the client's private value $x \in \F^L$ (where $\Ver(x) = 1$), 
the adversary learns
nothing about $x$ by participating in the SNIP protocol.
In particular, the adversary does not learn whether her guess
was correct, because her entire view can be simulated perfectly
without knowledge of~$x$. 
Condition~{(\ref{eq:snipzk})} implies a notion of 
semantic security: for any two valid $x_0, x_1 \in \F^L$, 
the adversary's view of the protocol is sampled identically, whether
the client uses $x_0$ or $x_1$ to execute the SNIP protocol.

Note that we only require the zero knowledge property to hold when 
$r \not \in \{1, \dots, M\}$.
Since we use a field $\F$ such that $|\F| \gg 2M$, the 
probability that a uniformly random $r$ falls 
into this bad set is negligibly small.
To entirely avoid this bad event, the servers could sample $r$ with
the constraint that $r \not\in \{1,\dots,M\}$, which would increase the soundness
error of the proof system slightly, to $1/(|\F| - M)$.

\medskip

To demonstrate that our SNIP construction satisfies the
zero-knowledge property, we must produce a simulator
$S = (S_1, S_2)$ that renders the adversary unable to 
distinguish the real and ideal worlds.

We first describe our simulator, and then argue that
the simulation is accurate.
\begin{framed}
  \noindent
\textsf{Simulator}.\ \ 
The simulator $S_1(r)$ first 
\begin{itemize}
  \item chooses random $(a, b, c) \gets_R \F^3$, subject to the
        constraint that $a \cdot b = c$,
  \item chooses $\Ash{x} \gets_R \F^L$,
  \item chooses $f(r),\ g(r) \gets_R \F$, 
  \item chooses $\Ash{f(0)},\ \Ash{g(0)} \gets_R \F$, and
  \item chooses $\Ash{h} \gets_R \F[x]$ of degree at most~$2M$ 
        by choosing coefficients at random from~$\F$. 
\end{itemize}
The simulator then sets:
\begin{align*}
  \Hsh{h(M)} &\gets (1 - \Ash{h}(M)) &&\in \F,\\
  h(r) &\gets f(r) \cdot g(r) &&\in \F, \text{ and}\\
  \Hsh{r \cdot h(r)} &\gets r \cdot (h(r) - \Ash{h}(r)) &&\in \F.
\end{align*}
The simulator computes shares
$(\Ash{a}, \Ash{b}, \Ash{c})$ of the multiplication triple,
shares $\Hsh{f(r)}$ and $\Hsh{g(r)}$, and the values
\begin{align*}
\Hsh{d} &= \Hsh{f(r)} - \Hsh{a} &&\in \F \\
\Hsh{e} &= \Hsh{r \cdot g(r)} - \Hsh{b} &&\in \F.
\end{align*}
Finally, the simulator produces the state
\[\state = \big(\, \Hsh{a}, \Hsh{b}, \Hsh{c}, \Hsh{r \cdot h(r)}\, \big).\]
and outputs
\begin{multline*}
\big(\state,\ (\Ash{a}, \Ash{b}, \Ash{c}, \Ash{x}, \\\Ash{f(0)}, \Ash{g(0)}, \Ash{h}, \Hsh{h(M)}, \Hsh{d}, \Hsh{e})\big).
\end{multline*}

\noindent
The simulator $S_2$, on input $(\textsf{state}, \Ash{\hat{d}}, \Ash{\hat{e}})$
computes $\hat{d} = \Ash{\hat{d}} + \Hsh{d}$ and 
$\hat{e} = \Ash{\hat{e}} + \Hsh{e}$, and 
outputs the unique value $\Hsh{\hat{\sigma}} \in \F$:
\[ \Hsh{\hat{\sigma}} \gets \hat{d}\hat{e}/2 + \hat{d}\Hsh{b} + \hat{e}\Hsh{a} + \Hsh{c} - \Hsh{r\cdot h(r)}. \]
\end{framed}

\begin{thm}
The adversary's view in the real world is distributed
identically to the adversary's view in the ideal world,
in which the adversary interacts with the simulator~$S$.
Therefore~{(\ref{eq:snipzk})} holds. 
\end{thm}
\begin{proof}
As a first step, notice that in the real world, the 
values $f(r)$ and $g(r)$ are distributed independently and
uniformly at random, conditioned on all the information that the
adversary has at the start of the experiment.
To see this, write the value $f(r)$ in terms
of the Lagrange interpolating polynomials $(\lambda_0(r), \dots, \lambda_M(r))$:
\begin{equation} \label{eq:lag}
 f(r) = \lambda_0(r) \cdot f(0) + \sum_{i=1}^M \lambda_i(r) \cdot f(i).
\end{equation}
Because the adversary knows $x$, it knows $f(1), \ldots, f(M)$.
However, in the real world, we
choose the value of $f(0) \in \F$ independently and uniformly at random.
Further, the value $\lambda_0(r)$ is non-zero for
all $r \not \in \{1, \dots, M\}$.
Therefore, by~{(\ref{eq:lag})} we have that $f(r)$ is uniform in $\F$
and is independent of the adversary's view. 
The same argument applies to $g(r)$.
Thus, for all relevant choices of $r$, $f(r)$ and $g(r)$ will be distributed uniformly
and independently of the adversary's view at the start of the experiment.
The simulator's choices of $f(r)$ and $g(r)$ are then distributed exactly as in
the real world.

\medskip

Next, notice that the simulation in the first flow of the experiment
of the shares of $a$, $b$, $c$, $x$, $f(0)$, $g(0)$, and $h$ is perfect, 
since the secret-sharing scheme perfectly hides these values in the real world.

The simulation of the share $\Hsh{h(M)}$ is also perfect.
Since $h(M)$ is the value on the output wire of the $\Ver$ circuit,
and since $\Ver(x) = 1$ for all honest clients, it will always
be true in the real world that $h(M) = 1$.
In particular, in the real world, the value $\Hsh{h(M)}$ that
the honest client sends satisfies
${\Ash{h}(M) + \Hsh{h(M)} = 1 \in \F}$ always.
The simulated value preserves this equality exactly.

The simulation of the shares of $d$ and $e$ is also perfect,
since $\Hsh{d}$ and $\Hsh{e}$ are masked by $\Hsh{a}$
and $\Hsh{b}$, which are equal to 
\[ \Hsh{a} = a - \Ash{a}\qquad\text{and}\qquad \Hsh{b} = b - \Ash{b},\]
for uniformly random values $a$ and $b$ unknown to the adversary.

\medskip

Now, the only remaining claim we need to argue is that the simulation of
the final value $\Hsh{\hat{\sigma}}$ is consistent with its value in the
real-world experiment.
To do so, we need to argue that, irrespective of the adversary's
choice of $\Ash{\hat{d}}$ and $\Ash{\hat{e}}$, and conditioned on all
other values the adversary knows, the simulated value of $\Hsh{\hat{\sigma}}$
is distributed as in the real world.
 
The values that the adversary should have sent, were it honest,
in the second step of the real-world experiment are: 
\[ \Ash{d} = \Ash{f(r)} - \Ash{a} \qquad \text{and} \qquad \Ash{e} = \Ash{r\cdot g(r)} - \Ash{b}.\]
Say that instead the adversary shifts these values by $\Delta_d$ and $\Delta_e$.
The reconstructed values $\hat{d}$ and $\hat{e}$ will then also be shifted:
\[ \hat{d} = d + \Delta_d \qquad \text{and} \qquad \hat{e} = e + \Delta_e.\]

\newcommand\hatsig{\Hsh{\hat{\sigma}}}
In this case, the value that the adversary receives in the real-world experiment
from the honest parties is:
\begin{align*}
  \hatsig &= \hat{d}\hat{e}/2 + \hat{d}\Hsh{b} + \hat{e}\Hsh{a} + \Hsh{c} - \Hsh{r \cdot h(r)}.
\end{align*}
Define the corresponding adversarial share:
\begin{equation}  \label{eq:simadv}
  \Ash{\hat{\sigma}} = \hat{d}\hat{e}/2 + \hat{d}\Ash{b} + \hat{e}\Ash{a} + \Ash{c} - \Ash{r \cdot h(r)}
\end{equation}

\begin{claim}
We have that
\begin{equation} \label{eq:sim}
 \Hsh{\hat{\sigma}} = f(r)\Delta_e + r\cdot g(r)\Delta_d + \Delta_d \Delta_e - \Ash{\hat{\sigma}}.
\end{equation}
\end{claim}
\begin{proof}[Proof of claim]
Consider the sum:
\begin{align*}
  \Ash{\hat{\sigma}} + \Hsh{\hat{\sigma}} &= \hat{d}\hat{e} + \hat{d}b + \hat{e}a + c - r\cdot h(r)\\
                              &= \begin{multlined}[t][5.5cm]
                                  (d + \Delta_d) (e + \Delta_e) + (d + \Delta_d)b \\ +(e + \Delta_e)a + ab - r\cdot h(r)
                                  \end{multlined}\\
                              &= \begin{multlined}[t][5.5cm]
                                  de + e\Delta_d + d\Delta_e + \Delta_e \Delta_d + db\\ + b\Delta_d + ea + a\Delta_e + ab - r\cdot h(r)
                                  \end{multlined}\\
                              &= \begin{multlined}[t][5.5cm]
                                  (de + db + ea + ab) + d\Delta_e + e\Delta_d \\+ \Delta_d \Delta_e + b\Delta_d + a\Delta_e - r\cdot h(r).
                                  \end{multlined}
                                \end{align*}
As in Beaver's protocol (Appendix~\ref{app:beaver:proto}), $r \cdot f(r)\cdot g(r) = (de + db + ea + ab)$. 
Additionally, $f(r) \cdot g(r) = h(r)$ for all honest clients, so
\begin{align*}
  \Ash{\hat{\sigma}} + \Hsh{\hat{\sigma}} &= \begin{multlined}[t][5.5cm]
                                  r\cdot f(r) g(r) + d\Delta_e + e\Delta_d \\+ \Delta_d \Delta_e + b\Delta_d + a\Delta_e - r\cdot h(r)
                                  \end{multlined}\\
                              &= d\Delta_e + e\Delta_d + \Delta_d \Delta_e + b\Delta_d + a\Delta_e \\
                              &= \begin{multlined}[t][5.5cm]
                                    (f(r) - a)\Delta_e + (r\cdot g(r) -b)\Delta_d \\+ \Delta_d \Delta_e + b\Delta_d + a\Delta_e
                                  \end{multlined}\\
                              &= f(r)\Delta_e + r\cdot g(r)\Delta_d + \Delta_d \Delta_e.
\end{align*}
Which proves~{(\ref{eq:sim})}.
\end{proof}

Now, 
the adversary knows $\Delta_d$ and $\Delta_e$ and, using the information that the adversary
has learned after the first flow in the experiment, the adversary can compute $\Ash{\hat{\sigma}}$ using~{(\ref{eq:simadv})}.
We know, by our arguments thus far, that
\begin{enumerate}[noitemsep,nolistsep]
  \item the adversary has no information about $f(r)$ and $g(r)$ at the start 
        of the experiment,
  \item the simulation of the first flow is perfect, and
  \item the contents of the first flow are independent of $f(r)$ and $g(r)$.
\end{enumerate}
These three points together demonstrate that whatever the adversary has learned
after seeing the first flow, it is independent of $f(r)$ and $g(r)$.
So, $\Delta_d$, $\Delta_e$, and $\Ash{\hat{\sigma}}$ are independent of $f(r)$ and $g(r)$,
and $f(r)$ and $g(r)$ are independent random values.

\noindent
In the real world, there are two cases:
\begin{itemize}
\item When $\Delta_d = \Delta_e = 0$, we have that $\Ash{\hat{\sigma}} = - \Hsh{\hat{\sigma}}$.

\item Otherwise, $f(r)\Delta_e + r\cdot g(r)\Delta_d$ is
a random value independent of everything else that the
adversary has seen so far. 
Therefore by~{(\ref{eq:sim})}, so is
$\Hsh{\hat{\sigma}}$.
\end{itemize}
Our simulation in the ideal world produces the same behavior.
To see this, observe that the simulator computes $\Hsh{\hat{\sigma}}$
exactly as in the real world, as long as all of the values upon 
which $\Hsh{\hat{\sigma}}$ depends 
are distributed as in the real world, and we have 
already argued that they are.
This completes the proof of the theorem.
\end{proof}

\paragraph{Why randomize the polynomials?}
If the client did not choose $f(0)$ and $g(0)$
at random in $\F$, that is, if the polynomials $f$ and $g$ were
completely determined by the input $x \in \F^L$, then the simulation
would fail.  The reason is that $f(r)$ and $g(r)$ would then be known
to the adversary (because the adversary knows $x$ in 
the real-world experiment), but the simulator (in the ideal world)
would not know $x$.  
As a result, when $\Delta_d \neq 0$ or $\Delta_e \neq 0$,
the simulator $S_2$ could not generate the
correct $\Hsh{\hat{\sigma}}$ to satisfy~{(\ref{eq:sim})}, 
and the simulation would fail.  
By having the client choose $f(0)$ and $g(0)$ at random we ensure
that $f(r)$ and $g(r)$ are independent of the adversary's view,
which enables the simulator to properly generate $\Hsh{\hat{\sigma}}$.

This simulation failure is not just an artifact of our proof technique:
if the SNIP used non-randomized polynomials $f$ and $g$, 
a malicious server could learn non-trivial information about $x$
by participating in the protocol.
The idea of the attack is that a malicious verifier could, 
by deviating from the protocol, learn whether $\Ver(\hat{x}) = 1$, 
for some value $\hat{x}$ related to $x$.
To do so, the adversary would essentially shift its share of $x$
to a share of $\hat{x}$ and then would execute the rest of the SNIP
protocol faithfully.
To compute the shifted share of $x$, the adversary 
would locally compute: $\Ash{\hat{x}} = \Ash{x} + \Delta_x$.
(The adversary also would need to shift its share of the polynomial $h$
to a share of $\hat{h}$, but this is also possible if the polynomials
$f$, $g$, and $h$ are not randomized.)
At the end of the SNIP verification protocol, the adversary learns
whether $\Ver(\hat{x}) = 1$, rather than whether $\Ver(x) = 1$.

Leaking $\Ver(\hat{x})$ to the adversary can potentially leak all of
$x$ to the adversary. 
For example, if $x$ is a single element of $\F$ and $\Ver(x)$ tests whether
$x \in \{0,1\}$, then learning $\Ver(\hat{x})$, for $\hat{x} = x+1$ leaks
$x$ entirely. That is, if ${x = 0}$, then 
\[\Ver(\hat{x}) = \Ver(x+1) = \Ver(1) = 1,\]
and if ${x = 1}$, then 
\[\Ver(\hat{x}) = \Ver(x+1) = \Ver(2) = 0.\]
An adversary who can learn $\Ver(\hat{x})$ can then completely recover~$x$.
A variant of the attack allows the adversary to test some related
predicate $\widehat{\Ver}(x)$ (instead of $\Ver(x)$) on the client's input,
which has essentially the same effect.

Randomizing the polynomials $f$ and $g$ prevents a malicious server
from learning side information about $x$ in this way.

 \section{Server-side $\Ver$ computation}
\label{app:serverside}

If the $\Ver$ predicate takes secret inputs from the servers,
the servers can compute $\Ver(x)$ on a client-provided input
$x$ without learning anything about $x$, except the value of $\Ver(x)$.
In addition, the client learns nothing about the $\Ver$ circuit, except the
number of multiplication gates in the circuit.

Crucially, this server-side variant only provides security against
``honest-but-curious'' servers, who execute the protocol correctly
but who try to learn everything possible about the clients' private data
from the protocol execution.
It may be possible to provide security against actively malicious servers using
a client-assisted variant of the SPDZ multi-party computation
protocol~\cite{SPDZ}, though we leave this extension to future work.

Let $M$ be the number of multiplication gates in the $\Ver$ circuit. 
To execute the $\Ver$ computation on the server side, the client sends
$M$ multiplication triple shares (defined in Appendix~\ref{app:beaver:proto}) 
to each server, along with a share of its private value $x$.
Let the $t$-th multiplication triple be of the form $(a_t, b_t, c_t) \in \F^3$.
Then define a circuit $\mathcal{M}$ that returns ``$1$'' if and only if
$c_t = a_t \cdot b_t$, for all $1 \leq t \leq M$.

The client can use a SNIP proof (Section~\ref{sec:disrupt:overview}) to
convince the servers that all of the $M$ triples it sent the servers are
well-formed.
Then, the servers can execute Beaver's multiparty computation protocol
(Section~\ref{app:beaver:proto}) to evaluate the circuit using the $M$
client-provided multiplication triples. 
The security of this portion of the protocol, with respect to ``honest but
curious'' servers, follows directly from the security of Beaver's protocol. 

Running the computation requires the servers to exchange $\Theta(M)$ field elements,
and the number of rounds of communication is proportional to the multiplicative depth
of the $\Ver$ circuit (i.e., the maximum number of multiplication gates on an
input-output path).

 \section{\afe definitions}
\label{app:afe-defs}

An \afe is defined relative to 
a field $\F$, two integers $k$ and $k'$ (where $k' \leq k$),
a set $\D$ of data elements, 
a set $\A$ of possible values of the aggregate statistic, 
and an aggregation function $f: \D^n \to \A$.
An \afe scheme consists of three efficient algorithms.  
The algorithms are:
\begin{itemize}

  \item $\Enc: \D \to \F^k$. 
        Covert a data item into its \afe-encoded
        counterpart.

  \item $\Ver: \F^k \to \{0,1\}$.
        Return ``$1$'' if and only if the input is in the image of $\Enc$. 

  \item $\Dec: \F^{k'} \to \A$.
        Given a vector representing a collection of encoded 
        data items, return the value of the aggregation function $f$
        evaluated at these items.

\end{itemize}

To be useful, an \afe encoding should satisfy the following 
properties:
\begin{defn}[\afe correctness]
We say that an \afe is {\em correct}
for an aggregation function $f$ if, 
for every choice of $(x_1, \dots, x_n) \in \D^n$, we have that:
\[ \Dec\big(\textstyle \sum_i \proj{k'}\big(\Enc(x_i)\big)\ \big) 
= f(x_1, \dots, x_n). \]
Recall that $\proj{k'}(v)$ denotes truncating the vector
$v \in \F_p^k$ to its first $k'$ components. 
\end{defn}

The correctness property of an \afe
essentially states that if we are given valid encodings of data items
$(x_1, \dots, x_n) \in \D^n$, the decoding of their sum should be 
$f(x_1, \dots, x_n)$.

\begin{defn}[\afe soundness]
We say that an \afe is {\em sound} if, for all encodings $e \in \F^k$:
the predicate $\Ver(e) = 1$ if and only if there exists a data item
$x \in \D$ such that $e = \Enc(x)$.
\end{defn}

An \afe is private with respect to a function $\hat{f}$, 
if the sum of encodings $\sigma = \sum_i
\proj{k'}(\Enc(x_i))$, given as input to algorithm
$\Dec$, reveals nothing about the underlying data beyond what 
$\hat{f}(x_1,\ldots,x_n)$ reveals.

\begin{defn}[\afe privacy]\label{defn:afepriv}
We say that an \afe is {\em private} with respect to a function
$\hat{f}:\D^n \to \A'$ if there exists an efficient simulator $S$
such that for all input data $(x_1,\ldots,x_n) \in \D^n$,
the distribution $S(\hat{f}(x_1,\ldots,x_n))$ is indistinguishable from the 
distribution $\sigma = \sum_i \proj{k'}(\Enc(x_i))$. 
\end{defn}

\paragraph{Relaxed correctness.}
In many cases, randomized data structures are more efficient than their
deterministic counterparts.
We can define a relaxed notion of correctness to capture a correctness
notion for randomized \afes.
In the randomized case, the scheme is parameterized by constants 
$0<\delta, \epsilon$ and the $\Dec$ algorithm
may use randomness.
We demand that with probability at least $1-2^{-\delta}$,
over the randomness of the algorithms, the encoding yields an ``$\epsilon$-good''
approximation of $f$.
In our applications, typically an $\epsilon$-good approximation
is within a multiplicative or additive factor of $\epsilon$ from the true value;
the exact meaning depends on the \afe in question.

 \section{Additional \afes}
\label{app:afe-extra}

\paragraph{Approximate counts.}
The frequency count \afe, presented in Section~\ref{sec:struct:basic},
works well when the client value $x$ lies in a small set of possible
data values $\D$.
This \afe requires communication linear in the size of $\D$.  
When the set $\D$ is large, a more 
efficient solution is to use a randomized counting data structure,
such as a count-min sketch~\cite{cormode2005improved}. 

Melis et al.~\cite{melis2016} demonstrated how to combine a count-min
sketch with a secret-sharing scheme to efficiently compute counts
over private data.
We can make their approach robust to malicious clients by 
implementing a count-min sketch \afe in \name.
To do so, we use $\ln (1/\delta)$
instances of the basic frequency count \afe, each for a set 
of size $e/\epsilon$,
for some constants $\epsilon$ and $\delta$, and where $e \approx 2.718$. 
With $n$ client inputs, the count-min sketch yields counts that are 
at most an additive $\epsilon n$
overestimate of the true values, except with probability $e^{-\delta}$.

Crucially, the $\Ver$ algorithm for this composed construction 
requires a relatively small number of multiplication gates---a few hundreds, 
for realistic choices of $\epsilon$ and $\delta$---so the servers can 
check the correctness of the encodings efficiently.

This \afe leaks the contents of a count-min sketch data structure
into which all of the clients' values $(x_1, \dots, x_n)$ have been inserted.

\itpara{Share compression.}
The output of the count-min sketch \afe encoding routine is essentially
a very sparse matrix of dimension $\ln(1/\delta) \times (e/\epsilon)$.
The matrix is all zeros, except for a single ``1'' in each row.
If the \name client uses a conventional secret-sharing scheme to split
this encoded matrix into $s$ shares---one per server---the 
size of each share would be as large as the matrix itself, even though
the plaintext matrix contents are highly compressible.

A more efficient way to split the matrix into shares 
would be to use a function secret-sharing
scheme~\cite{boyle2016function,boyle2015function,gilboa2014distributed}.
Applying a function secret sharing scheme to each row of the encoded
matrix would allow the size of each share to grow as the square-root of
the matrix width (instead of linearly).
When using \name with only two servers, there are very efficient 
function secret-sharing constructions that would allow the shares
to have length logarithmic in the width of the matrix~\cite{boyle2016function}.
We leave further exploration of this technique to future work.

\paragraph{Most popular.}
Another common task is to return the most popular string in a data
set, such as the most popular homepage amongst a set of Web clients. 
When the universe of strings is small, it is
possible to find the most popular string using the frequency-counting
\afe. When the universe is large (e.g., the set of all URLs), this method is not
useful, since recovering the most popular string would require
querying the structure for the count of every possible string.
Instead, we use a simplified version of a data structure of Bassily
and Smith~\cite{bassily2015local}.

When there is a very popular string---one that more than $n/2$ clients
hold, we can construct a very efficient \afe for collecting it.  
Let $\F$ be a field of size at least~$n$. 
The $\Enc(x)$ algorithm represents its input $x$ as a $b$-bit string 
$x = (x_0, x_1, x_2, \dots, x_{b-1}) \in \{0,1\}^b$, 
and outputs a vector of
$b$ field elements $(\beta_0, \dots, \beta_{b-1}) \in \F^b$, where
$\beta_i = x_i$ for all $i$.  
The $\Ver$ algorithm uses $b$ multiplication gates to check
that each value $\beta_i$ is really a $0/1$ value in $\F$, as in 
the summation \afe. 

The $\Dec$ algorithm gets as input the sum of $n$ such encodings 
$\sigma = \sum_{i=1}^n \Enc(x_i) = (e_0, \dots, e_{b-1}) \in \F^b$.
The $\Dec$ algorithm rounds each value $e_i$ either down to zero or up to $n$ (whichever is closer)
and then normalizes the rounded number by $n$ to get a $b$-bit 
binary string $\sigma \in \{0,1\}^b$, which it outputs.
As long as there is a string $\sigma^*$ with popularity greater than $50\%$, 
this \afe returns it.
To see why, consider the first bit of $\sigma$. 
If $\sigma^*[0]=0$, then the sum $e_0 < n/2$ and $\Dec$ outputs ``$0$.''
If $\sigma^*[0]=1$, then the sum $e_0 > n/2$ and $\Dec$ outputs ``$1$.''
Correctness for longer strings follows

This \afe leaks quite a bit of information about the given data.
Given $\sigma$, one learns the number of data values that have their
$i$th bit set to 1, for every $0 \leq i < b$.  In fact, the \afe is
private relative to a function that outputs these $b$ values, which
shows that nothing else is leaked by $\sigma$.

With a significantly more complicated construction, we can adapt a
similar idea to collect strings that a constant fraction $c$
of clients hold, for $c \leq 1/2$.  The idea is to have the
servers drop client-submitted strings at random into different
``buckets,'' such that at least one bucket has a very popular string
with high probability~\cite{bassily2015local}.

\paragraph{Evaluating an arbitrary ML model.}
We wish to measure how well a public regression model 
predicts a target $y$ from a client-submitted feature vector~$x$. 
In particular, if our model outputs a prediction $\hat{y}=M(x)$, we would like
to measure how good of an approximation $\hat{y}$ is of $y$.
The $R^2$ coefficient is one statistic for capturing this information.

Karr et al.~\cite{karr2005secure} observe that it is possible to 
reduce the problem of computing the $R^2$ coefficient of a public regression
model to the problem of computing private sums. 
We can adopt a variant of this idea to use \name to compute the $R^2$
coefficient in a way that leaks little beyond the coefficient itself.

The $R^2$-coefficient of the model for client inputs $\{x_1,\ldots,x_n\}$ is
$R^2 = 1 - \sum_{i=1}^2 (y_i - \hat{y}_i)^2 / \text{Var}(y_1,\ldots,y_n)$,
where $y_i$ is the true value associated with $x_i$, 
$\hat{y}_i = M(x_i)$ is the predicted value of $y_i$, 
and $\text{Var}(\cdot)$ denotes variance. 

An \afe for computing the $R^2$ coefficient works as follows.
On input $(x,y)$, the $\Enc$ algorithm first computes 
the prediction $\hat{y} = M(x)$ using the public model $M$.
The $\Enc$ algorithm then outputs the 
tuple $(y, y^2, (y - \hat{y})^2, x)$, embedded in a 
finite field large enough to avoid overflow.

Given the tuple $(y, Y, Y^*, x)$ as input, the $\Ver$ algorithm 
ensures that $Y = y^2$ and $Y^* = (y - M(x))^2$.
When the model $M$ is a linear regression model,
algorithm $\Ver$ can be represented as an arithmetic circuit
that requires only two multiplications.
If needed, we can augment this with a check that the $x$ values 
are integers in the appropriate range using a range check, as in 
prior \afes.
Finally, given the sum of encodings restricted to the first three
components, the $\Dec$ algorithm 
has the information it needs to compute the $R^2$ coefficient.

This \afe is private with respect to a function that outputs the
$R^2$ coefficient, along with the expectation and variance of 
$\{y_1,\ldots,y_n\}$. 

 \section{\Name protocol and proof sketch}
\label{app:security}

We briefly review the full \name protocol and then discuss its security. 

\paragraph{The final protocol.}
We first review the \name protocol from Section~\ref{sec:struct}.
Let there be $m$ malicious clients whose values are included in the final aggregate.
We assume that every honest client~$i$, for $i \in \{1, \dots, n-m\}$,
holds a private value $x_i$ that lies in some set of data items $\D$. 
We want to compute an aggregation function $f: \D^{n-m} \to \A$ on these 
private values using an \afe.
The \afe encoding algorithm $\Enc$ maps $\D$ to $\F^k$, for some field $\F$
and an arity $k$. When decoding, the encoded vectors in $\F^k$ are first
truncated to their first $k'$ components. 

The \name protocol proceeds in four steps:

\begin{enumerate}[leftmargin=\parindent]
\item \textbf{Upload.}
      Each client $i$ computes $y_i \gets \Enc(x_i)$ and splits its encoded
      value into $s$ shares, one per server.
      To do so, the client picks random values $[y_i]_1, \dots, [y_i]_s \in \F^k$,
      subject to the constraint:
      $y_i = [y_i]_1 + \cdots + [y_i]_s \in \F^k$.
      The client then sends, over an encrypted and authenticated channel, one share 
      of its submission to each server, along with a share of a SNIP
      proof (Section~\ref{sec:disrupt}) that $\Ver(y_i) = 1$.

\item \textbf{Validate.} \label{sec:struct:proto:valid}
      Upon receiving the $i$th client submission, the servers 
      verify the client-provided SNIP
      to jointly confirm that $\Ver(y_i)=1$ 
      (i.e., that client's submission is well-formed).
      If this check fails, the servers reject the submission.

\item \textbf{Aggregate.}
      Each server $j$ holds an accumulator value $A_j \in \F^{k'}$, 
      initialized to zero, where $0 < k' \leq k$.
      Upon receiving a share of a client encoding $[y_i]_j \in \F^k$,
      the server truncates $[y_i]_j$ to its first $k'$ components, and
      adds this share to its accumulator: 
      $$A_j \gets A_j + \proj{k'}([y_i]_j) \in \F^{k'}.$$
      Recall that $\proj{k'}(v)$ denotes truncating the vector
      $v \in \F_p^k$ to its first $k'$ components.

\item \textbf{Publish.} \label{step:publish}
      Once the servers have received a share from each client, they publish
      their accumulator values. 
      The sum of the accumulator
      values $\sigma = \sum_j A_j \in \F^{k'}$ yields the 
      sum $\sum_i \proj{k'}(y_i)$ of the clients' private
      encoded values.
      The servers output $\Dec(\sigma)$. 
\end{enumerate}

\paragraph{Security.} 
We briefly sketch the security argument for the complete protocol. 
The security definitions appear in Appendix~\ref{app:secdefs}.

First, the robustness property (Definition~\ref{def:robust}) 
follows from the soundness of the
SNIP construction: as long as the servers are honest, they will 
correctly identify and reject any client submissions that do 
not represent proper \afe encodings.

Next, we argue $f$-privacy (Definition~\ref{defn:fpriv}). 
Define the function 
$$g(x_1, \dots, x_{n-m}) = \sum_{i=1}^{n-m} \proj{k'}\big(\Enc(x_i)\big).$$ 
We claim that, as long as:
\begin{itemize}
  \item at least one server executes the protocol correctly, 
  \item the \afe construction is private with respect to $f$,
        in the sense of Definition~\ref{defn:afepriv}, and
  \item the SNIP construction satisfies the zero-knowledge
        property (Section~\ref{sec:disrupt:overview}),
\end{itemize}
the only information that leaks to the adversary is 
the value of the function $f$ on the private values of 
the honest clients included in the final aggregate.

To show this, it suffices to construct a simulator $S$ that takes as input
$\sigma = g(x_1, \dots, x_{n-m})$ and outputs a transcript of the protocol
execution that is indistinguishable from a real transcript.
Recall that the \afe simulator takes $f(x_1,\ldots,x_{n-m})$ as input
and simulates $\sigma$.  
Composing the simulator $S$ with the \afe simulator yields
a simulator for the entire protocol, 
as required by Definition~\ref{defn:fpriv}.

On input $\sigma$, the simulator $S$ executes these steps:
\begin{itemize}[leftmargin=\parindent]
  \item To simulate the submission of each honest client, 
        the simulator invokes the SNIP simulator as a subroutine. 
  \item To simulate values produced by adversarial clients, the simulator
        can query the adversary (presented to the simulator as an oracle) 
        on the honest parties' values generated so far.
  \item The simulator must produce a simulation of the
        values produced by the honest servers 
        in Step~\ref{step:publish} of the protocol.

        Let $\sigma$ be the sum of the honest clients' encodings. 
        For the simulation to be accurate, the honest servers must publish
        values $A_j$ such that the honest servers' accumulators sum
        to (1) all of the shares of encodings given to the honest servers
        by adversarial clients plus (2) $\sigma$ minus (3) 
        the shares of honest clients' simulated encodings given to adversarial servers.

        Let $[y_i]_j$ be the $j$th share of the encoding sent by client $i$.
        Let $\textsf{Serv}_\A$ be the set of indices of the adversarial servers and
        $\textsf{Client}_\A$ be the set of indices of the adversarial clients.
        Let $\textsf{Serv}_H$ and $\textsf{Client}_H$ be the set of indices of 
        the honest servers and clients, respectively.
        Then the accumulators $A_j$ must satisfy the relation:
        \begin{multline*}
          \sum_{j \in \textsf{Serv}_H} A_j = \sigma + \left(\sum_{j \in \textsf{Serv}_H}\sum_{i \in \textsf{Client}_\A}[y_i]_j\right) \\- \left(\sum_{j \in \textsf{Serv}_\A}\sum_{i \in \textsf{Client}_H} [y_i]_j\right).
                  \end{multline*}
        The simulator can pick random values for the honest servers' $A_j$s subject to 
        this constraint, since the simulator knows $\sigma$, it knows the values that 
        the honest clients' sent to the adversarial servers, and it knows
        the values that the adversarial clients sent to the honest servers.
\end{itemize}
To argue that the simulator $S$ correctly simulates a 
run of the real protocol: 
\begin{itemize}[leftmargin=\parindent]
  \item The existence of the SNIP simulator implies that everything
        the adversarial servers see in the SNIP checking step, when 
        interacting with an honest client, is independent of the
        honest client's private value $x_i$, provided that $\Ver(x_i) = 1$.
        This property holds even if the malicious servers deviate from
        the protocol in a way that causes an honest client's submission
        to be rejected by the honest servers.
  \item The simulation of the aggregation step of the protocol is perfect.
        In the real protocol, since the adversary sees only $s-1$ shares of
        each client submission, the values $A_j$ are just 
        random values subject to the constraint above.
\end{itemize}

\medskip
Finally, anonymity (Definition~\ref{def:anon}) follows by
Claim~\ref{claim:anonpriv} whenever the function $f$ is symmetric.
Otherwise, anonymity is impossible, by Claim~\ref{claim:nonsym}.

 \section{Additional optimizations}
\label{app:opts}

\paragraph{Optimization: PRG secret sharing.}
The \name protocol uses additive secret sharing to split the clients' private data
into shares.
The na\"ive way to split a value $x \in \F^L$ 
into $s$ shares is to choose $[x]_1, \dots, [x]_{s-1} \in \F^L$
at random and then set $[x]_s = x - \sum_{i=1}^{s-1} [x]_i \in \F^L$.
A standard bandwidth-saving optimization is to generate the first $s-1$ shares
using a pseudo-random generator (PRG) $G: \K \to \F^L$,
such as AES in counter mode~\cite{hugo93,BR07}.
To do so, pick $s-1$ random PRG keys $k_1, \dots, k_{s-1} \in \K$, 
and define the first $s-1$ shares as $G(k_1), G(k_2), \dots, G(k_{s-1})$.
Rather than representing the first $s-1$ shares as vectors in $\F^L$, 
we can now represent each of the first $s-1$ shares using a single AES key.
(The last share will still be $L$ field elements in length.)
This optimization 
reduces the total size of the shares from $sL$ field elements down to $L + O(1)$.
For $s=5$ servers, this $5\times$ bandwidth savings is significant.

\paragraph{Optimization: Verification without interpolation.}
In the course of verifying a SNIP
(Section~\ref{sec:disrupt:snip}, Step~\ref{proto:poly}), 
each server $i$ needs to interpolate two large polynomials
$[f]_i$ and $[g]_i$. 
Then, each server must evaluate the polynomials $[f]_i$ and $[g]_i$ at a
randomly chosen point $r \in \F$ to execute
the randomized polynomial identity test (Step~\ref{proto:test}).

The degree of these polynomials is close to $M$, where $M$ is the number of multiplication 
gates in the $\Ver$ circuit.
If the servers used straightforward polynomial interpolation and evaluation to verify the
SNIPs, the servers would need to perform $\Theta(M \log M)$
multiplications to process a single client submission, even using optimized FFT methods.
When the $\Ver$ circuit is complex (i.e., $M \approx 2^{16}$ or more),
this $\Theta(M \log M)$ cost will be substantial.

Let us imagine for a minute that we could fix {\em in advance}
the random point $r$ that the servers use to execute
the polynomial identity test.
In this case, each server can perform interpolation and
evaluation of any polynomial $P$ 
in one step using only $M$ field multiplications per server,
instead of $\Theta(M\log M)$.
To do so, each server precomputes constants $(c_0, \dots, c_{M-1}) \in \F^{M}$.
These constants depend on the $x$-coordinates of the points being interpolated
(which are always fixed in our application) and on the point $r$
(which for now we assume is fixed).
Then, given points $\{(t, y_t)\}_{t=0}^{M-1}$ on a polynomial $P$, the servers can evaluate
$P$ at $r$ using a fast inner-product computation: $P(x) = \sum_t c_t y_t \in \F$.
Standard Lagrangian interpolation produces
these $c_i$s as intermediate values~\cite{interp}.

Our observation is that the servers {\em can} fix the ``random'' point $r$ 
at which they evaluate the polynomials $[f]_i$ and $[g]_i$ as long as: 
(1) the clients never learn $r$, and
(2) the servers sample a new random point $r$ periodically.
The randomness of the value $r$ only affects soundness.
Since we require soundness to hold only if all \name servers are honest, we
may assume that the servers will never reveal the value $r$ to the clients.

A malicious client may try to learn something over time about the servers'
secret value $r$ by sending a batch of well-formed and malformed submissions and
seeing which submissions the servers do or do not accept.
A simple argument shows that after making $q$ such queries, the client's 
probability of cheating the servers is at most $(2M+1)q/|\F|$.
By sampling a new point after every $Q \approx 2^{10}$ client uploads, the servers can 
amortize the cost of doing the interpolation precomputation over $Q$ client uploads, 
while 
keeping the failure probability bounded above by $(2M+1)Q/|\F|$, which they might
take to be $2^{-60}$ or less.

In \name, we apply this optimization to combine the interpolation of $[f]_i$
and $[g]_i$ with the evaluation of these polynomials at the point $r$.

In Step~\ref{proto:poly} of the SNIP verification process, 
each server must also evaluate the
client-provided polynomial $[h]_i$ at each point $t \in \{0, \dots, M\}$.
To eliminate this cost, we have the client send the polynomial $[h]_i$ 
to each server $i$ in point-value form.
That is, instead of sending each server shares of the coefficients of $h$, 
the client sends each server shares of evaluations of $h$. 
In particular, the client evaluates $[h]_i$ at all of 
the points $t \in \{0, \dots, 2M\}$ and
sends the evaluations $[h]_i(0), [h]_i(1), [h]_i(2), \dots, [h]_i(2M)$ to the server.
Now, each server $i$ already has the evaluations of $[h]_i$ at all of the points it
needs to complete Step~\ref{proto:poly} of the SNIP verification.
To complete Step~\ref{proto:test}, each server must interpolate
$[h]_i$ and evaluate $[h]_i$ at the point $r$.
We accomplish this using the same fast interpolation-and-evaluation trick
described above for $[f]_i$ and $[g]_i$

\paragraph{Circuit optimization}
In many cases, the servers hold multiple verification 
circuits $\Ver_1, \dots, \Ver_N$ and want to check whether 
the client's submission passes all $N$ checks.
To do so, we have the $\Ver$ circuits return zero (instead of one) on success.
If $W_j$ is the value on the last output wire of the circuit $\Ver_j$,
we have the servers choose random values $(r_1, \dots, r_N) \in \F^N$
and publish the sum $\sum_j r_j W_j$ in the last step of the protocol.
If any $W_j \neq 0$, then this sum will be non-zero with high probability
and the servers will reject the client's submission.

\end{document}